\def\ARXIVVERSION{1}
\def\PAPERWITHSUPPLEMENT{1}
\newif\ifrefereelayout
\refereelayoutfalse   % Public copy: single-spaced journal layout
\ifrefereelayout
  \documentclass[AEJ,reviewmode]{AEA}
\else
  \documentclass[AEJ,finalmode]{AEA}
\fi
\ifdefined\ARXIVVERSION
  \newcommand{\paperversionurl}{https://arxiv.org/pdf/2604.00874}
\else
  \newcommand{\paperversionurl}{https://cashman.science/assets/pdf/conditional-disclosure.pdf}
\fi
\newif\ifanonymizeforpeerreview
\anonymizeforpeerreviewfalse   % Named
\newif\ifbuildwithsupplement
\ifdefined\PAPERWITHSUPPLEMENT
  \buildwithsupplementtrue
\else
  \buildwithsupplementfalse
\fi
\usepackage{natbib}
\usepackage{amsfonts}
\usepackage{amssymb}
\usepackage{bm}
\usepackage{booktabs}
\usepackage{array}
\usepackage{graphicx}
\usepackage{refcount}
\usepackage{xr}
\usepackage{subfiles}
\usepackage{xcolor}
\definecolor{linkblue}{RGB}{31,78,140}
\usepackage[colorlinks=true,linkcolor=linkblue,citecolor=linkblue,urlcolor=linkblue]{hyperref}
\usepackage{doi}
\ifrefereelayout
  \usepackage{aej-review-12pt}
\fi

\makeatletter
\long\def\XR@test#1#2#3#4\XR@{%
  \let\XR@tempa\@gobbletwo
  \ifx#1\newlabel
    \let\XR@tempa\@firstoftwo
  \else\ifx#1\@input
    \let\XR@tempa\@secondoftwo
  \fi\fi
  \XR@tempa{#1{\XR@prefix#2}{#3}}{\edef\XR@list{\XR@list#2\relax}}%
  \ifeof\@inputcheck\expandafter\XR@aux
  \else\expandafter\XR@read\fi}
\makeatother

\makeatletter
\let\AC@original@appsect\@appsect
\def\@appsect#1#2#3#4#5#6[#7]#8{%
  \NR@gettitle{#7}%
  \AC@original@appsect{#1}{#2}{#3}{#4}{#5}{#6}[{#7}]{#8}%
}
\g@addto@macro\appendix{\let\p@subsection\@empty}
\makeatother

\makeatletter
\def\ps@headings{%
  \leftskip=0pt\rightskip=0pt
  \let\@oddfoot\@empty
  \let\@evenfoot\@empty
  \def\@oddhead{{\small\textit{\MakeUppercase{\hfil\@shortTitle\hfil\thepage}}}}%
  \def\@evenhead{{\small\textit{\MakeUppercase{\thepage\hfil\@shortTitle\hfil}}}}%
}
\makeatother

\newtheorem{assumption}{Assumption}

\newtheorem{lemma}{Lemma}
\newtheorem{proposition}{Proposition}
\newtheorem{theorem}{Theorem}
\newtheorem{corollary}{Corollary}
\newtheorem{remark}{Remark}

\newcommand{\E}{\mathbb{E}}
\newcommand{\Prob}{\mathbb{P}}

\newcommand{\ind}{\mathbf{1}}
\newcommand{\appendixparttitle}[1]{%
  \clearpage
  \begin{center}
  {\scshape\sectionsize #1}
  \end{center}
  \vspace{12pt}
}

\ifSubfilesClassLoaded{%
  \externaldocument{conditional_disclosure}%
}{%
  \ifbuildwithsupplement\else
    \externaldocument{conditional_disclosure_online_appendix}%
  \fi
}

\newcommand{\xref}[1]{\ifbuildwithsupplement\ref{#1}\else\ref*{#1}\fi}
\newcommand{\xeqref}[1]{\ifbuildwithsupplement\eqref{#1}\else\textup{(\ref*{#1})}\fi}

\providecommand{\theHsection}{}
\providecommand{\theHsubsection}{}
\providecommand{\theHequation}{}
\providecommand{\theHfigure}{}
\providecommand{\theHtable}{}
\providecommand{\theHassumption}{}
\providecommand{\theHdefinition}{}
\providecommand{\theHlemma}{}
\providecommand{\theHproposition}{}
\providecommand{\theHtheorem}{}
\providecommand{\theHcorollary}{}
\providecommand{\theHremark}{}
\newcommand{\startsupplement}{%
  \appendix
  \setcounter{section}{1}%
  \setcounter{subsection}{0}%
  \counterwithin*{assumption}{section}%
  \counterwithin*{definition}{section}%
  \counterwithin*{lemma}{section}%
  \counterwithin*{proposition}{section}%
  \counterwithin*{theorem}{section}%
  \counterwithin*{corollary}{section}%
  \counterwithin*{remark}{section}%
  \renewcommand{\theassumption}{\thesection\arabic{assumption}}%
  \renewcommand{\thedefinition}{\thesection\arabic{definition}}%
  \renewcommand{\thelemma}{\thesection\arabic{lemma}}%
  \renewcommand{\theproposition}{\thesection\arabic{proposition}}%
  \renewcommand{\thetheorem}{\thesection\arabic{theorem}}%
  \renewcommand{\thecorollary}{\thesection\arabic{corollary}}%
  \renewcommand{\theremark}{\thesection\arabic{remark}}%
  \renewcommand{\theHsection}{supplement.\arabic{section}}%
  \renewcommand{\theHsubsection}{supplement.\arabic{section}.\arabic{subsection}}%
  \renewcommand{\theHequation}{supplement.\arabic{section}.\arabic{equation}}%
  \renewcommand{\theHfigure}{supplement.\arabic{section}.\arabic{figure}}%
  \renewcommand{\theHtable}{supplement.\arabic{section}.\arabic{table}}%
  \renewcommand{\theHassumption}{supplement.\arabic{section}.\arabic{assumption}}%
  \renewcommand{\theHdefinition}{supplement.\arabic{section}.\arabic{definition}}%
  \renewcommand{\theHlemma}{supplement.\arabic{section}.\arabic{lemma}}%
  \renewcommand{\theHproposition}{supplement.\arabic{section}.\arabic{proposition}}%
  \renewcommand{\theHtheorem}{supplement.\arabic{section}.\arabic{theorem}}%
  \renewcommand{\theHcorollary}{supplement.\arabic{section}.\arabic{corollary}}%
  \renewcommand{\theHremark}{supplement.\arabic{section}.\arabic{remark}}%
}

\draftSpacing{1.5}

\begin{document}

\title{Conditional Disclosure as a Coordination Device}
\shortTitle{Conditional Disclosure as a Coordination Device}
\issueName{}
\ifanonymizeforpeerreview
\author{Anonymous}
\else
\author{Matthew Cashman\thanks{Cashman: MIT Sloan School of Management, cashman@mit.edu. I wish to thank Adam Bear, Rahul Bhui, Caspar Kaiser, Drazen Prelec, Birger Wernerfelt, Zachary Wojotowicz, and Erez Yoeli for many helpful comments and suggestions. I am especially indebted to Alex Tabarrok for his substantial feedback.}\\[6pt]
\normalfont\normalsize \today\ (\href{\paperversionurl}{latest version here})}
\fi

\begin{abstract}
People with controversial views may stay silent rather than speak out alone, though in a large enough group they could speak safely together. A Social Assurance Contract makes disclosure conditional: signed commitments of support stay private until enough people have signed for it to be safe to publish. In contrast, a survey can measure hidden support but does not make identification conditional on enough others joining. I derive the contract's exact advantage over speaking openly, and when that advantage disappears. Such contracts can turn a few people who would speak alone into a larger coalition.
\end{abstract}


\maketitle

People often privately support controversial ideas, yet stay silent in public. They may be willing to join a large public coalition were there one to join, but unwilling to be identified in a small one even when they know that others share their view. This situation differs from pluralistic ignorance, in which people misjudge others' private views \citep{prenticePluralisticIgnoranceAlcohol1993,prenticePluralisticIgnorancePerpetuation1996}. Preference falsification offers a related account of why private views and public statements diverge \citep{kuranPrivateTruthsPublic1997}. Measuring hidden support does not by itself solve this coordination problem, so people need a way to make their names public only if enough others do so with them.

A \emph{Social Assurance Contract} provides that guarantee. Its architect fixes the statement, who is eligible to sign, and the condition for publication. People submit endorsements: signed permission to publish that exact statement under their names. The administrator holds those endorsements privately and publishes the names together only if the condition is met. Otherwise, it publishes no names. The architect and administrator may be the same person or a cryptographic system that requires trusting no single party. A university contract might read:

\begin{quote}
\textsc{We, the undersigned, believe [controversial belief] and think the University should take the following steps [1, 2, 3, \ldots]. Signatures on this letter will become public \textit{only when at least [200] faculty have signed.} Until then, no one can see who has signed. If the threshold is not met, the letter and its signatures are destroyed.}
\end{quote}

Anonymous measurement, open expression, and conditional release do different jobs. A survey or secret ballot reports an aggregate without naming respondents, while what I call a \emph{credible poll} is an authenticated version. This version verifies endorsements from a specified eligible population and publishes a certified count rather than names. These institutions work when someone can act on the aggregate, but they do not assemble a coalition willing to stand publicly behind a statement. Open expression creates that coalition but names participants even when few people join. A Social Assurance Contract keeps endorsements in private custody and jointly publishes them only when the resulting coalition is large enough. I study the risk this publication condition removes: being exposed after participation remains low and coordination fails.

Social Assurance Contracts apply the threshold logic of step-level public goods to identity disclosure. Provision-point mechanisms refund contributions when a funding threshold fails \citep{bagnoliProvisionPublicGoods1989}, dominant assurance contracts add refund bonuses \citep{tabarrokPrivateProvisionPublic1998}, and \citet{strauszTheoryCrowdfundingMechanism2017} studies provision-point crowdfunding. Here the mechanism holds endorsements rather than money, and its closest institutional antecedent is the information escrow of \citet{ayresInformationEscrows2012}, which stores allegations and releases them under specified corroboration conditions. \citet{braghieriThresholdDisclosureCollective2026} also study threshold identity disclosure, but let each voter choose when her own vote--identity pair becomes visible. A Social Assurance Contract instead releases one named coalition as a unit.

The model focuses on material retaliation: costs that the opposition, such as an employer or professional group, imposes on a named participant. Open expression and the contract begin with the same group committed to participation, the \emph{vanguard}. People outside that group are \emph{marginal agents}, whose decision to join can depend on the institution. They have the same preferences but receive different private signals about how favorable conditions are for a coalition. Under either institution, retaliation attaches to membership in the published coalition and weakens in the same way as the coalition grows. This cost channel is consistent with evidence that lower retaliation costs increase employee complaints \citep{heeseEffectRetaliationCosts2021}. Holding the audience's beliefs fixed isolates the effect of the publication condition and the cost of signing the contract. Social-image models offer another reason why visible isolation is costly \citep{bernheimTheoryConformity1994,braghieriPoliticalCorrectnessSocial2024}, but the comparison here holds those inferences fixed. Garbling protects individual whistleblowers in \citet{chassangCrimeIntimidationWhistleblowing2019}, whereas conditional release protects them by delaying identification until a coalition forms.

Vanguard participation is committed before the marginal agents choose to join or not. Its members, who may be protected or already public, express under the open institution and submit endorsements under the contract, where their commitment need not yet be public. This timing distinguishes the simultaneous benchmark here from the companion paper's sequential recruitment, where the vanguard is already public. In favorable states, the vanguard alone clears the publication threshold and provides enough protection to make joining worthwhile even after signing costs. This gives the selection argument its high-state starting point, called an upper dominance region. Conditional release amplifies an existing vanguard into a larger coalition, but only if that vanguard is credible and the threshold is reachable.

The main result predicts behavior when people have nearly the same information about conditions but remain uncertain about what others know. It gives agents private signals with small errors and applies the same reasoning to both institutions. Under each institution, agents then participate in states above a cutoff, and the contract helps when its cutoff is lower; in the paper's example, the contract's cutoff state is $0.618$ and open expression's is $0.643$. To measure the difference, I hold the state fixed and ask how much expressive benefit each institution requires. The resulting comparison averages payoffs as marginal participation varies from none to all. A \emph{rank} is a point on this conceptual path: the fraction of marginal agents who participate, not their arrival order in a campaign. The \emph{rank-indifference value} is the value of public expression that makes joining break even on average across these ranks. The two institutions produce different values because the contract keeps signers private below the threshold.

Their difference, the \emph{accounting gap}, is the advantage from removing exposure between the vanguard and threshold, less the signing cost adjusted for the chance of publication. Signers pay even on failure, so the benefit they receive on success must cover that cost. The less likely publication is, the larger this cost adjustment becomes. Theorem~\ref{thm:selection_gap} also requires the threshold to protect the people the contract recruits: being named in a coalition that just reaches the threshold must be at least as valuable to them as remaining unnamed. The \emph{selected gap} is therefore the smaller of the accounting gap and the recruitment-coverage margin, the amount by which the open-expression cutoff exceeds the value needed for that protection. Within the theorem's domain, it measures how far the contract lowers the value of public expression needed to participate. In a single economy, the contract lowers the cutoff state exactly when the accounting gap is positive at open expression's cutoff state. The gain is positive when avoided exposure outweighs adjusted signing costs and a protective, reachable threshold can recruit people who would stay silent under open expression.

The selection argument follows the global-games approach \citep{carlssonGlobalGamesEquilibrium1993,morrisGlobalGamesTheory2003,frankelEquilibriumSelectionGlobal2003}. Agents first rule out actions that lose whatever others do. They then rule out actions that lose given what others have already ruled out, and continue this reasoning. As the errors in private signals shrink, the surviving choices approach a common boundary, above which agents participate and below which they stay out. I apply this procedure to both institutions, whose complete-information versions can have several equilibria.

Global games have been used to study bank runs \citep{goldsteinDemandDepositContracts2005}, targeted coordination subsidies \citep{sakovicsWhoMattersCoordination2012}, discriminatory success-contingent rewards \citep{winterIncentivesDiscrimination2004,halacRaisingCapitalHeterogeneous2020}, and information manipulation \citep{edmondInformationManipulationCoordination2013}. Conditional disclosure changes whether a participant is identified after failure, and the accounting gap measures the benefit of that change. The result requires strategic complementarity; turnout can instead be substitutable, as in the Hong Kong evidence of \citet{cantoniProtestsStrategicGames2019}.

Anonymous measurement offers a different form of protection. Randomized response uses noise to measure sensitive attitudes \citep{warnerRandomizedResponseSurvey1965}, and recent work examines respondents' incentives under this mechanism. \citet{blumeElicitingPrivateInformation2019} show that noise can support truthful reporting by giving respondents plausible deniability, but can also produce informative equilibria with systematic misreporting. \citet{fisherDesigningRandomizedResponse2023} show that, when respondents' traits are correlated through an unknown population rate, a larger sample makes any one answer less important for beliefs about its author. The researcher can then use less noise while preserving truthful reporting. These mechanisms protect respondents by limiting what others can infer from their answers. A Social Assurance Contract instead publishes an exact named coalition when it is large enough to be safe.

Information about social norms can also change behavior \citep{bursztynMisperceivedSocialNorms2020}, but learning that others agree does not make identification conditional on enough others joining. At the same state and audience-belief index, simultaneous open expression can sustain either vanguard-only or full expression. Holding public information fixed lets the comparison measure the gain from withholding signer identities after failure. Publication may also create common knowledge, depending on who observes it and how they communicate \citep{chweStructureStrategyCollective1999}, a separate effect held fixed here.

The design results explain why the publication rule takes the form of a threshold. Consider rules that depend only on coalition size and publish either every endorsement or none. A threshold at the minimum safe size releases the list at every size that protects its signers and withholds it at smaller sizes. Within this class, no safe rule can release more often, and a protective threshold also preserves the incentive to join as participation rises. To deliver that protection, an independent administrator or a protocol must keep endorsements private and verify that the threshold has been reached before releasing names. If failed lists are published with the same exposure and expressive value as open expression, the contract instead reproduces open expression plus its signing cost.

Two implementation threats require different responses. The contract may be false, collecting names to leak them; signers then need a minimum level of trust before they sign. People who sign in bad faith, or ``infiltrators'', may also recant, repudiating the statement immediately after publication. When only continuing supporters protect one another, a credible bound on recantation determines how many extra names a safe release needs. I derive the resulting trust threshold and recantation margin.

The Supplemental Appendix also studies a finite group, where one endorsement can determine whether the threshold is reached. Agents begin with prior beliefs and receive private signals, while signing costs and publication risk shape their choices around a vanguard that clears the threshold in favorable states.

A companion paper asks which coalition can safely be published when signers differ in their protection. Starting from a set $V$ of valid endorsements, it finds a safely releasable coalition $D(V)$ \citep{cashmanSafeCollectiveExpression2026}. It studies how the identities of coalition members affect safety, when recruitment can proceed without risk, and how risk tolerance and leakage change the coalitions that can form. This paper asks what induces people to enter $V$ in the first place, in a setting where marginal agents share one payoff type, receive different private signals, and face the same audience under the contract and simultaneous open expression. The papers address complementary questions under different assumptions; neither model is a special case of the other.

Section~\ref{sec:conditional_coalition_model} defines the institution and model. Section~\ref{sec:mechanism_comparison} derives the participation comparison and its implications for contract design. Section~\ref{sec:applications} applies the model to collective whistleblowing, named organizing committees, and public letters with a credible vanguard, a reachable threshold, and exposure that falls as the final coalition grows, and states testable implications. Section~\ref{sec:discussion} concludes, and Appendix A contains proofs. The Supplemental Appendix develops the finite benchmark, separates failure privacy from attempt privacy, states the informational-exchangeability lemma, derives the trust threshold and recantation margin, and describes an implementation architecture.

\section{Conditional Public Coalition Formation}\label{sec:conditional_coalition_model}

A larger final coalition lowers the cost of being named, and the comparison asks how this safety in numbers changes the decision to join. To isolate that effect, marginal agents have the same preferences and face the same exposure environment in both institutions. I approximate a large population by a continuum, treating each individual as too small to change coalition size alone even though their choices collectively determine it. The Supplemental Appendix returns to a finite group, where one signature can be pivotal.

\subsection{The Finite Institution}

Private support, signing, and public membership are three different statuses. A supporter may never sign, and a signer becomes public only if the contract meets its publication condition. The finite institution makes these distinctions explicit.

Let $N$ be a fixed finite eligible population. The supporter set $S\subseteq N$ contains the people willing to be published if the published coalition protects them. A supporter has not necessarily signed. Let $V\subseteq S$ be the people whose valid endorsements the administrator has received. Everyone has unit mass in this finite description, so a coalition $C$ has mass $m(C)=|C|$. Given an integer publication threshold $K$, the release rule is
\begin{equation}
D_K(V)=
\begin{cases}
V, & |V|\ge K,\\
\varnothing, & |V|<K.
\end{cases}
\end{equation}
Thus the contract jointly publishes the complete signer set when the publication condition is met, and the resulting public coalition is $C=D_K(V)$. Otherwise, it publishes no names.

Failure privacy means that $D_K(V)=\varnothing$ when $|V|<K$, so no signer's identity is published. That definition does not by itself say whether outsiders observe the campaign or learn that it failed. The ideal benchmark adds \emph{attempt privacy} and withholds the failed count, so the audience imposing exposure costs observes
\begin{equation}
o=
\begin{cases}
V, & |V|\ge K,\\
\varnothing, & |V|<K,
\end{cases}
\end{equation}
where $\varnothing$ is a public non-event in which the audience sees no names, count, or failure announcement. The Supplemental Appendix studies the distinction between failure privacy and attempt privacy.

In the strategic finite benchmark of the Supplemental Appendix, the vanguard signs regardless of the marginal agents' choices, and the marginal agents use one signal cutoff: each signs when her private signal exceeds a common boundary. A single marginal endorsement can determine whether total endorsements reach $K$. Higher signing friction raises the cutoff and weakly lowers publication probability in every state. The continuum model below provides the comparison with simultaneous open expression.

\subsection{The Vanguard Analytical Benchmark}

How far participation can expand beyond the vanguard depends on two quantities that vary with the state: the vanguard's size and the number of additional supporters available to join.

The state $\theta\in\mathbb{R}$ summarizes conditions that make a public coalition more or less viable, such as latent support or the availability of protected participants. The vanguard has mass $b(\theta)$ and participates under every institution. The total available mass is $Y(\theta)$, so $Y(\theta)-b(\theta)$ is the marginal pool whose participation can change with the mechanism. The normalization
$0\le b(\theta)<Y(\theta)\le1$ expresses all coalition sizes as shares of the relevant population. Both $b$ and $Y$ are continuously differentiable and weakly increasing, so more favorable states never shrink either the vanguard or the available pool.

The marginal agents form an atomless population $\mathcal{I}=[0,1]$ with normalized mass one: each individual has zero mass, and together they fill the unit interval. They have the same preferences. Their expressive benefit $e$ is the value of having their name publicly attached to the statement, and $\alpha\ge0$ measures their sensitivity to exposure. If a fraction $w\in[0,1]$ participates, the candidate coalition has mass
\begin{equation}
y(w;\theta)=b(\theta)+\bigl(Y(\theta)-b(\theta)\bigr)w.
\end{equation}
At $w=0$, only the vanguard participates, so candidate mass is $b(\theta)$. At $w=1$, everyone available joins, so candidate mass is $Y(\theta)$. The unit interval thus measures the fraction of the marginal pool that joins, while coalition mass measures its size in the full population. Intermediate values describe different possible participation levels, not a sequence in time. Under simultaneous open expression, every participant enters the public coalition. Under the contract, the administrator holds the corresponding signer set $V$ in private custody and publishes $C=D(V)$ only when the publication condition is met. I use $Y$ for participation capacity and $y$ for candidate coalition mass.

The audience-belief index $\mu$ summarizes the public environment facing a signer, including beliefs formed from prior information or updated after a survey. For example, it can describe how sympathetic an employer or professional audience is. The signer takes these beliefs as given. The function $\ell(y,\mu)$ is the baseline exposure cost of appearing in a coalition of mass $y$, and a signer with sensitivity $\alpha$ pays $\alpha\ell(y,\mu)$. Exposure is bounded, continuous, nonnegative, and strictly decreasing in $y$. This last condition, safety in numbers, means that adding names lowers each signer's cost. A more favorable audience-belief index also lowers exposure. I hold $\mu$ fixed across institutions to isolate the protection that a larger coalition supplies even when the audience's beliefs stay the same. The Supplemental Appendix discusses the separate case with $\mu(y)$, where beliefs change with coalition size. Under material retaliation, a larger coalition spreads or constrains the punishment directed at each member.

Marginal agents share the type $(e,\alpha)$ but have different information. Agent $i$ observes a private signal
\begin{equation}
x_i=\theta+\sigma\eta_i.
\end{equation}
The noise term $\eta_i$ creates disagreement across agents, and $\sigma>0$ controls its size. Conditional on $\theta$, the shocks are independent and identically distributed with a continuous, strictly increasing distribution $F_\eta$ and positive density $f_\eta=F_\eta'$. The signals satisfy the monotone likelihood-ratio property, under which a higher signal shifts relative likelihood toward higher states.

A pure strategy is a jointly measurable map $\varphi:\mathcal{I}\times\mathbb{R}\to\{0,1\}$, where $\varphi_i(x)=1$ means that agent $i$ joins after signal $x$. Profiles are ordered pointwise. The benchmark imposes \emph{conditional exact aggregation}, the continuum version of a law of large numbers: at each state, the realized participation share equals the share predicted by the signal distribution,
\begin{equation}
w_\varphi(\theta)=\int_{\mathcal{I}}\int_{\mathbb R}\varphi_i(\theta+\sigma\eta)\,dF_\eta(\eta)\,di.
\end{equation}
Under a symmetric cutoff strategy $\varphi_i(x)=\mathbf 1\{x\ge c\}$, exact aggregation becomes
\begin{equation}
w_c(\theta)=1-F_\eta\!\left(\frac{c-\theta}{\sigma}\right).
\end{equation}
This is both the probability that a signal drawn in state $\theta$ exceeds the cutoff and the realized participation share in the continuum. Exact aggregation supports the selection analysis, while the Supplemental Appendix uses a finite population. Public information, including $\mu$, stays fixed as private noise $\sigma$ vanishes.

Any two coalitions of the same size create the same exposure cost, regardless of who belongs to them. This is the assumption of \emph{exchangeable identities}. Size determines the cost; identities determine who bears it. When seniority, subgroup, or network position changes protection, the membership of the coalition also matters, and the companion paper studies that setting.

\subsection{Benchmark Mechanisms}

The institutions differ in what they collect and what they make public (Table~\ref{tab:mechanism_classification}). Anonymous measurement can change the audience's beliefs before people decide whether to participate, but it need not verify or hold publishable endorsements. A credible poll does both, yet publishes only a certified count. Neither releases a named coalition.

\begin{table}[t]
\caption{Mechanism Classification}
\label{tab:mechanism_classification}
\centering
\small
\begin{tabular}{@{}p{0.19\textwidth}p{0.23\textwidth}p{0.30\textwidth}p{0.14\textwidth}@{}}
\toprule
Mechanism & Public output & Identity disclosure & Private custody \\
\midrule
Simultaneous open expression & Named participants & Every participant is named & No \\
Anonymous measurement & Aggregate estimate or vote tally & No respondent is directly named & Not necessarily \\
Credible poll & Authenticated count & No signer is directly named & Yes \\
Social Assurance Contract & Named coalition & All signers jointly on success; none on failure & Yes \\
\bottomrule
\end{tabular}
\end{table}

Under simultaneous open expression, marginal agents choose without observing other marginal choices, and every participant is named. The companion paper instead studies sequential public organizing, where later participants can observe earlier choices. A Social Assurance Contract receives endorsements privately, publishes the complete signer set together when its mass reaches $\tau$, and otherwise publishes no names. Signing costs $k\ge0$, and the baseline assumes that the administrator credibly follows this custody and release rule.

\section{Mechanism Comparison}\label{sec:mechanism_comparison}

To isolate what changes when identification is conditional on a large enough coalition, I hold the vanguard, audience beliefs, private signals, and selection procedure fixed. A named participant's exposure depends on coalition size in the same way under either institution. I call the institutions the two \emph{arms} of the comparison. A survey changes behavior through $\mu$ or beliefs about the state; conditional disclosure changes what happens to a participant after coordination fails.

\subsection{Selection, Not Information}\label{subsec:selection_not_information}

The complete-information benchmark reveals the coordination problem that private signals will resolve. Even when everyone observes the state and audience, each person's best action depends on how many others participate.

\begin{proposition}[Open-expression multiplicity]\label{prop:no_selection}
Fix a state $\theta$ and any audience belief $\mu$. In the simultaneous open-expression game among marginal agents with a participating vanguard of mass $b(\theta)$, the gain from expression rises with the expressing mass. If the marginal type satisfies
\begin{equation}
\alpha\,\ell\bigl(Y(\theta),\mu\bigr)\ \le\ e\ <\ \alpha\,\ell\bigl(b(\theta),\mu\bigr),
\end{equation}
then both the vanguard-only profile, in which no marginal agent expresses, and the full-expression profile are equilibria. The band is nondegenerate whenever $\alpha>0$, $b(\theta)<Y(\theta)$, and $\ell$ is strictly decreasing in $y$. An informational intervention can move $\mu$ and move a type into or out of this band. At every fixed state and index inside it, the open-expression game remains multiple.
\end{proposition}

Different expectations sustain the two equilibria. If every marginal agent stays silent, a lone deviation leaves one person exposed beside only the vanguard, so deviation does not pay. If everyone speaks, the large coalition lowers exposure enough to make speaking worthwhile. The state and audience are identical; only expectations about others differ.

Small differences in private signals select between these outcomes. Under open expression, a participant always receives expressive benefit $e$ and bears exposure. Under the contract, benefit and exposure arise only when the coalition reaches $\tau$, while the signer pays friction $k$ whether the contract succeeds or fails:
\begin{equation}
\begin{gathered}
u_{open}(\theta,w)=e-\alpha\,\ell\bigl(y(w;\theta),\mu\bigr),\\
u_{SAC}(\theta,w)=\mathbf{1}\{y(w;\theta)\ge\tau\}\bigl(e-\alpha\,\ell(y(w;\theta),\mu)\bigr)-k .
\end{gathered}
\end{equation}
The cost $k$, paid only under the contract, covers the signer's extra effort of verification, entrusting her endorsement to custody, and establishing trust in the platform. Open expression has no such added cost in the baseline, but participants still bear the exposure cost shown in its payoff.

The reasoning starts with decisions that do not depend on anyone else's choices. At low enough signals, staying out must be best even if everyone else joins, and at high enough signals, joining must be best even if no other marginal agent joins. The assumption below establishes these two dominance regions and the regularity conditions needed to connect them. Its first six clauses require failure in bad states, success supplied by the vanguard in good states, protection at the threshold, and a positive cost of signing even when failure is certain.

\begin{assumption}[Vanguard-based selection domain]\label{ass:operative_domain}
There are states $\theta_L<\theta_H$ and an interior evaluation state $\theta\in(\theta_L,\theta_H)$ with $b(\theta)<\tau<Y(\theta)$ such that:
\begin{enumerate}
\item \emph{$\tau$-protection:} $e\ge\alpha\,\ell(\tau,\mu)$ --- a coalition at the threshold protects the marginal type. The threshold must protect whom it recruits.
\item \emph{Low-state exposure:} $e<\alpha\,\ell(Y(\theta_L),\mu)$ --- at low states even full participation does not protect, so the type faces a coordination problem.
\item \emph{Genuine failure risk:} $\tau>Y(\theta_L)$ --- the threshold is unreachable at low states.
\item \emph{Vanguard clearance:} $b(\theta_H)\ge\tau$ --- at high states the threshold is reachable by the vanguard alone.
\item \emph{Vanguard protectability with friction:} $e>\alpha\,\ell(b(\theta_H),\mu)+k$ --- at high states the vanguard alone makes participation worthwhile.
\item \emph{Strict contract friction:} $k>0$, which makes nonparticipation strictly dominant when contract failure is certain.
\end{enumerate}
For each arm $r\in\{open,SAC\}$, the rank integral
\begin{equation}
R_r(s)=\int_0^1u_r(s,w)\,dw
\end{equation}
averages the participation payoff over all ranks $w\in[0,1]$ at state $s$. The integral is continuous and strictly single crossing on the relevant state range, with a unique zero $\theta_r^*\in(\theta_L,\theta_H)$.

The signal family satisfies the monotone-likelihood-ratio property, and public information stays fixed as $\sigma$ approaches zero. Beliefs after a signal can take either of two forms. In the diffuse calculation, an agent with signal $x$ represents the state as $\theta=x-\sigma\eta$, where $\eta\sim F_\eta$, and this kernel corresponds to an improper flat prior. Alternatively, the state has a proper common prior with bounded support and continuously differentiable density $g$ that is strictly positive on an open interval containing $[\theta_L,\theta_H]$. Both forms produce the same small-noise rank calculation.
\end{assumption}

The rank integral asks whether joining pays on average as marginal participation moves from none to full. The economic clauses also supply the footholds for iterative reasoning: Appendix~\ref{app:proofs} shows that, once noise is small, sufficiently low signals make staying out a strict best response and sufficiently high signals make joining one, whatever others do. The type-dependent clauses can all hold only if the signing cost is smaller than the improvement in exposure between the unfavorable and favorable benchmark states:
\begin{equation}
k\ <\ \alpha\bigl(\ell(Y(\theta_L),\mu)-\ell(b(\theta_H),\mu)\bigr)
\end{equation}
along with the other displayed inequalities. Favorable conditions must improve safety enough to cover the contract's cost. This restricts each one-type economy and does not come from unmodeled differences among agents.

\begin{remark}[A primitive regularity example]\label{rem:primitive_selection_example}
The following example satisfies the entire selection domain. Let the state lie in $[0,1]$ with a continuously differentiable positive prior density, let signal errors be standard normal, and set
\begin{equation}
\begin{gathered}
b(s)=\frac{1}{10}+\frac{3}{5}s,
\qquad
Y(s)=\frac{1}{5}+\frac{4}{5}s,
\qquad
\ell(y,\mu)=1-y,\\
\tau=\frac{3}{5},
\qquad
\alpha=1,
\qquad
e=\frac{2}{5},
\qquad
k=\frac{1}{50}.
\end{gathered}
\end{equation}
With these choices, direct integration gives
\begin{equation}
R_{open}(s)=-\frac{9}{20}+\frac{7}{10}s
\end{equation}
and
\begin{equation}
R_{SAC}(s)=
\begin{cases}
-\dfrac{1}{50}, & 0\le s\le\dfrac{1}{2},\\[6pt]
\dfrac{4(2s-1)^2}{5(2s+1)}-\dfrac{1}{50},
& \dfrac{1}{2}<s<\dfrac{5}{6},\\[8pt]
-\dfrac{47}{100}+\dfrac{7}{10}s,
& \dfrac{5}{6}\le s\le1.
\end{cases}
\end{equation}
When $s\le1/2$, even the largest available coalition falls short of the threshold, so the contract fails at every rank. For $1/2<s<5/6$, success requires enough marginal agents to join. At $s=5/6$ and above, the vanguard meets the threshold by itself, so the contract succeeds at every rank. The average open-expression payoff crosses zero once, at $9/14$. The contract payoff crosses zero once, at $(81+\sqrt{321})/160$, strictly between $1/2$ and $5/6$ and below $9/14$.

The example also supplies the dominance states from which the selection argument starts. At $s=1/10$, open participation pays $2/5-\ell(Y(1/10),\mu)=-8/25<0$ even if every marginal agent joins, while the contract cannot reach the threshold and pays $-1/50$. At $s=9/10$, the vanguard is $16/25>3/5$, so participation pays at least $1/25$ under open expression and $1/50$ under the contract even if no other marginal agent joins. The inequalities remain strict near both states. Because participation payoffs rise with the state and higher signals point to higher states, sufficiently low signals make staying out best whatever others do, while sufficiently high signals make joining best. The example therefore supplies both dominance regions and a unique zero for each average payoff. The general result remains conditional on Assumption~\ref{ass:operative_domain}.
\end{remark}

The vanguard matters because it makes joining worthwhile without help from other marginal agents in favorable states. If $b(s)<\tau$ at every state, one infinitesimal signer cannot trigger publication but still pays $k$, so no signal makes signing profitable regardless of others' choices. The contract's selection argument therefore starts with an existing vanguard, not an empty coalition.

Starting from the dominance regions, agents repeatedly rule out an action that loses against every behavior others could still choose. Appendix~\ref{app:proofs} defines the process at every agent and signal, allowing strategies that differ across agents or respond nonmonotonically to signals. As noise vanishes, the signal interval where surviving strategies can disagree shrinks to a point. This is \emph{essential uniqueness}: behavior is pinned down away from the boundary, even though the complete-information game has several equilibria.

\begin{lemma}[Noise-limit selection]\label{lem:global_selection}
Under Assumption~\ref{ass:operative_domain}, fix either arm $r\in\{open,SAC\}$. As $\sigma\downarrow0$, the lower and upper signal boundaries of the strategies surviving iterated deletion of interim strictly dominated actions both converge to $\theta_r^*$, the unique zero of $R_r$. Rationalizable strategies can disagree only between these boundaries. Equivalently, for every $\delta>0$, sufficiently small $\sigma$ forces participation at every $x>\theta_r^*+\delta$ and nonparticipation at every $x<\theta_r^*-\delta$.

At the fixed evaluation state $\theta$, define an arm's \emph{rank-indifference value} $\bar e_{arm}(\theta)$ by
\begin{equation}
\int_0^1 u_{arm}(\theta,w)\,dw=0 .
\end{equation}
When the boundary economy satisfies Assumption~\ref{ass:operative_domain}, the rank-indifference value is the arm's selected expressive-benefit cutoff: the marginal type participates in otherwise identical one-type economies above the cutoff and stays out below it.
\end{lemma}

The rank-indifference value is the private benefit from public expression that makes the marginal type break even after averaging over the ranks generated by the global game. Interpret it by comparing populations that differ only in their common expressive benefit $e$, not by placing several types in one population. Appendix~\ref{app:proofs} proves the lemma for all strategies surviving iterated deletion.

The lemma makes the comparison tractable because, in the small-noise limit, an agent at the participation boundary weights every fraction of other marginal participants equally, the Laplacian belief of \citet{morrisGlobalGamesTheory2003}. Each arm then has a cutoff state, the zero of its rank integral, and participation occurs above it; in Remark~\ref{rem:primitive_selection_example} the contract's cutoff is $0.618$ and open expression's is $0.643$. Comparing cutoff states does not generally yield a simple benefit-minus-cost formula, because each is the root of an institution-specific integral. Holding the state fixed and asking how much expressive benefit each institution requires instead yields an exact benefit-minus-cost formula, and Corollary~\ref{cor:cutoff_states} translates it back into cutoff states. The next theorem performs this calculation and then checks that the publication threshold protects the people it would recruit.

\begin{theorem}[The selection gap]\label{thm:selection_gap}
At an interior evaluation state, where the vanguard is below the threshold and full participation is above it, let
\begin{equation}
w_\tau=\frac{\tau-b(\theta)}{Y(\theta)-b(\theta)}\in(0,1)
\end{equation}
be the fraction of the marginal pool needed to move the coalition from $b(\theta)$ to $\tau$. The contract fails below $w_\tau$ and publishes at or above it. The two rank-indifference values are
\begin{equation}
\begin{gathered}
\bar e_{open}(\theta)=\alpha\int_0^1 \ell\bigl(y(w;\theta),\mu\bigr)dw,\\
\bar e_{SAC}(\theta)=\alpha\,\frac{\int_{w_\tau}^1 \ell\bigl(y(w;\theta),\mu\bigr)dw}{1-w_\tau}+\frac{k}{1-w_\tau},
\end{gathered}
\end{equation}
and their \emph{accounting gap} is
\begin{equation}
\begin{aligned}
G^{acct}(\theta)
\equiv\bar e_{open}(\theta)-\bar e_{SAC}(\theta)
={}&
\alpha\,w_\tau\operatorname{avg}_{[0,w_\tau]}\ell\bigl(y(w;\theta),\mu\bigr)\\
&-\alpha\,w_\tau\operatorname{avg}_{[w_\tau,1]}\ell\bigl(y(w;\theta),\mu\bigr)\\
&-\frac{k}{1-w_\tau}.
\end{aligned}
\end{equation}
Lemma~\ref{lem:global_selection} applies only when threshold protection, $e\ge\alpha\ell(\tau,\mu)$, holds. Combining that requirement with rank indifference gives the contract's expressive-benefit requirement
\begin{equation}
e^{sel}_{SAC}(\theta)
=\max\{\bar e_{SAC}(\theta),\alpha\ell(\tau,\mu)\}.
\end{equation}
Whenever the separate one-type economies invoked in the comparison satisfy the remaining clauses of Assumption~\ref{ass:operative_domain}, the selected participation-cutoff gap is
\begin{equation}
G^{sel}(\theta)
=\bar e_{open}(\theta)-e^{sel}_{SAC}(\theta)
=\min\{G^{acct}(\theta),\bar e_{open}(\theta)-\alpha\ell(\tau,\mu)\}.
\end{equation}
A positive $G^{sel}$ is the width, in units of $e$, of the interval of types within the theorem's domain that join the failure-private contract but not open expression.
\end{theorem}

Open expression averages exposure over the entire path because every participant is named at every rank. The contract averages exposure only over successful ranks, since below the threshold the signer is neither named nor given the value of public expression. The formula for $\bar e_{SAC}$ adds $k/(1-w_\tau)$ because the signer pays $k$ at every rank but receives expressive value only on the successful share $1-w_\tau$. Dividing by that share gives the expressive value needed on success to cover the always-paid cost.

In the accounting gap, the first two lines compare exposure below and above the threshold. Exposure falls as the coalition grows, so the below-threshold average is larger. Their weighted difference is the \emph{exposure-tail advantage}: the exposure between the vanguard and threshold that failure privacy removes. The last line subtracts the friction haircut, stated in the same expressive-benefit units. The accounting gap records the payoff effect of withholding signer identities after failure, before the selection conditions are imposed. When $G^{acct}>0$, avoided exposure is worth more than the contract's cost.

Accounting alone does not determine selection. Threshold protection ensures that publication at the threshold does not make joining less attractive. The minimum in $G^{sel}$ reflects two tests: the rank accounting must favor the contract, and the threshold must protect the people it recruits. Protection binds when $\bar e_{SAC}(\theta)<\alpha\ell(\tau,\mu)$, that is, when
\begin{equation}
k<\alpha(1-w_\tau)\Bigl(\ell(\tau,\mu)-\operatorname{avg}_{[w_\tau,1]}\ell\bigl(y(w;\theta),\mu\bigr)\Bigr).
\end{equation}
The bracket is positive because exposure keeps falling above the threshold, so protection binds whenever signing friction is small enough. The selected gap is then $\bar e_{open}(\theta)-\alpha\ell(\tau,\mu)$, which does not depend on $k$; for larger friction, $G^{sel}=G^{acct}$. A type with $e<\alpha\ell(\tau,\mu)$ faces a contract payoff that falls at the threshold before rising, and may still welcome publication in a larger coalition. Selection with such nonmonotone payoffs requires different methods \citep{goldsteinDemandDepositContracts2005} and is left open.

The theorem compares the institutions in units of expressive benefit at a fixed state. Within a single economy, the same comparison determines which institution has the lower cutoff state.

\begin{corollary}[Cutoff states]\label{cor:cutoff_states}
Fix one economy that satisfies Assumption~\ref{ass:operative_domain} in both arms, and suppose the open-expression cutoff state is interior, $b(\theta^*_{open})<\tau<Y(\theta^*_{open})$. Then $\theta^*_{SAC}<\theta^*_{open}$ if and only if $G^{acct}(\theta^*_{open})>0$.
\end{corollary}

At its own cutoff state, the open-expression type is exactly rank indifferent, $e=\bar e_{open}(\theta^*_{open})$, so the contract's rank integral there has the sign of $\bar e_{open}-\bar e_{SAC}$, and single crossing places the contract's cutoff below open expression's exactly when that difference is positive. In Remark~\ref{rem:primitive_selection_example}, the accounting gap at $9/14$ is $3/175$, and the contract's cutoff falls to $(81+\sqrt{321})/160\approx0.618$. The corollary reads the theorem within one economy; the theorem's interval of types reads the same comparison across economies at a fixed state.

The size of the advantage also depends on which costs the two institutions share. The baseline treats verification and escrow costs as specific to the contract. If public expression requires the same costly action, then
\begin{equation}
\bar e_{open}^{sym}(\theta)=\bar e_{open}(\theta)+k,
\end{equation}
while $\bar e_{SAC}^{sym}(\theta)=\bar e_{SAC}(\theta)$. The friction haircut becomes $k w_\tau/(1-w_\tau)$ rather than $k/(1-w_\tau)$ because open expression now also costs $k$. The protection cap applies under either normalization. Use the baseline when verification and custody create the incremental cost, and the symmetric formula when both forms of expression require the same costly action.

Figure~\ref{fig:selection_gap} illustrates the baseline decomposition and the two regimes. In panel A, the shaded area is the exposure difference between the below-threshold ranks and the above-threshold average; signing friction is then subtracted separately. Panel B varies the signing cost: for $0<k<0.05625$ protection binds and the gap covered by the theorem stays at $0.10$ while the accounting gap rises, and above it the two coincide.

\begin{figure}[t]
\centering
\includegraphics[width=\textwidth]{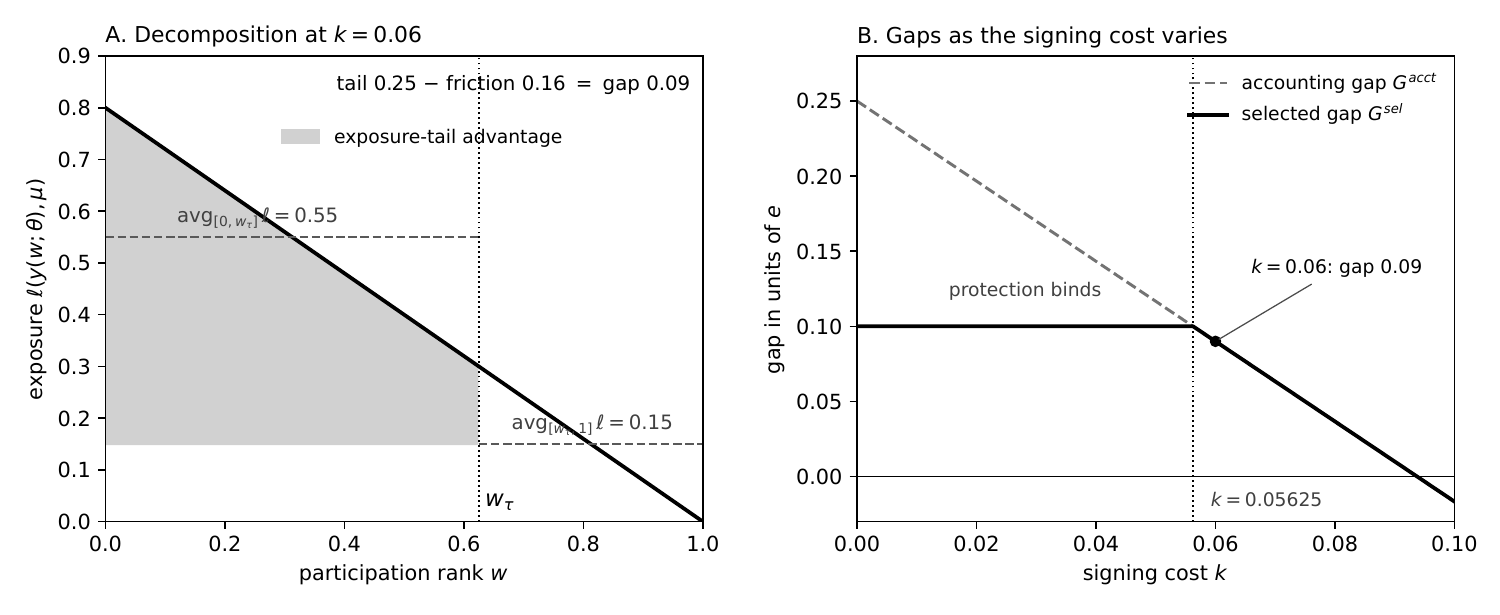}
\caption{Exposure-tail accounting and the protection cap. Panel A: exposure $\ell(y(w;\theta),\mu)$ follows the recruitment path for $\ell(y)=1-y$, $b(\theta)=0.2$, $Y(\theta)=1$, $\tau=0.7$, $\alpha=1$, and $k=0.06$, so $w_\tau=0.625$. The dashed lines mark the below- and above-threshold rank averages, $0.55$ and $0.15$. The shaded exposure-tail advantage is the scaled difference $\alpha w_\tau(0.55-0.15)=0.25$. The friction haircut is $k/(1-w_\tau)=0.16$, leaving $G^{acct}=0.09$; because $\bar e_{SAC}=0.31\ge\alpha\ell(\tau,\mu)=0.30$, threshold protection does not bind and $G^{sel}=G^{acct}$. Panel B: the accounting gap $0.25-k/0.375$ and the selected gap $\min\{G^{acct},\bar e_{open}-\alpha\ell(\tau,\mu)\}=\min\{G^{acct},0.10\}$, the gap covered by the theorem's domain, as the signing cost varies. Protection binds for $0<k<0.05625$, where the selected gap is $0.10$ whatever the signing cost; the accounting gap reaches zero at $k=0.09375$. Selection requires $k>0$, so $k=0$ appears only as a limit.}
\label{fig:selection_gap}
\end{figure}

The contract is most useful for people who want to speak but need a larger group to make doing so worthwhile. The next corollary identifies this region and the cases where it disappears.

\begin{corollary}[Conditional comparison across one-type economies]\label{cor:gap_properties}
Fix the common structural objects and consider separate one-type economies indexed by $(e,\alpha)$. Statements about selected behavior require the relevant economy to satisfy Assumption~\ref{ass:operative_domain}, including protection, dominance, vanguard clearance, and friction. Parts 2 and 3 cover two boundaries: coalition size does not lower exposure, and the vanguard already reaches the threshold.
\begin{enumerate}
\item \emph{Intermediate region.} In any such economy whose marginal type satisfies
\begin{equation}
\max\{\bar e_{SAC}(\theta),\ \alpha\ell(\tau,\mu)\}\ <\ e\ <\ \bar e_{open}(\theta)
\end{equation}
the marginal type joins under the selected contract but not selected open expression. Across one-type economies, the strict interval has width $G^{sel}(\theta)>0$. Its weak recruitment-coverage boundary is
\begin{equation}
\bar e_{open}(\theta)\ \ge\ \alpha\,\ell(\tau,\mu)
\qquad\text{(equivalently } \textstyle\int_0^1 \ell(y(w;\theta),\mu)\,dw \ \ge\ \ell(\tau,\mu)\text{)}
\end{equation}
with strict inequality for a nonempty interval. The monotone-selection argument requires $e\ge\alpha\ell(\tau,\mu)$; otherwise publication makes the type's payoff jump downward.
\item \emph{Boundary without safety in numbers.} If $\ell$ is constant in coalition mass on the relevant range, failure privacy still removes exposure on failed ranks, but also withholds expressive value. There is no exposure-tail advantage, and the accounting gap is $-k/(1-w_\tau)<0$.
\item \emph{Vanguard-at-threshold null.} If $b(\theta)\ge\tau$ at the evaluation state --- the limit $w_\tau\downarrow0$ with $w_\tau=\max\{0,(\tau-b)/(Y-b)\}$ --- the tail vanishes and the accounting gap is exactly $-k$. When the vanguard already publishes by itself, the contract protects no additional below-threshold ranks and adds only friction. Its value comes from removing exposure between the vanguard and the threshold.
\item \emph{Fixed-index scope.} At a common fixed $\mu$, the exposure-tail advantage is positive exactly when average exposure below the threshold exceeds average exposure above it. The accounting gap is positive exactly when this advantage exceeds $k/(1-w_\tau)$, and the selected gap also requires positive recruitment coverage. Information can change $\mu$ and hence both cutoffs and gaps. Holding $\mu$ fixed isolates conditional disclosure: failure privacy prevents below-threshold exposure. Supplemental Appendix~\xref{app:privacy_exchangeability} discusses the separate inference channel.
\end{enumerate}
\end{corollary}

\begin{corollary}[Common-rank accounting robustness]\label{cor:common_rank}
At an interior evaluation state, let $W$ have any distribution $H$ on $[0,1]$, applied identically to both institutions. Some parts of the conceptual recruitment path can receive more weight than others, but the two arms must use the same weights. Define
\begin{equation}
L(W)=\ell\bigl(y(W;\theta),\mu\bigr),
\qquad
p_H=\Prob_H(W\ge w_\tau)\in(0,1).
\end{equation}
The corresponding $H$-weighted accounting values are
\begin{equation}
\bar e_{open}^H=\alpha\E_H[L(W)],
\qquad
\bar e_{SAC}^H
=\alpha\E_H[L(W)\mid W\ge w_\tau]+\frac{k}{p_H},
\end{equation}
and their difference is
\begin{equation}
\begin{aligned}
\bar e_{open}^H-\bar e_{SAC}^H
={}&\alpha(1-p_H)
\left(
\E_H[L(W)\mid W<w_\tau]
-\E_H[L(W)\mid W\ge w_\tau]
\right)\\
&-\frac{k}{p_H}.
\end{aligned}
\end{equation}
Under strict safety in numbers, average exposure is higher below the threshold, so every common rank weighting retains the same avoided-tail-minus-friction accounting structure.
\end{corollary}

Corollary~\ref{cor:gap_properties} identifies when failure privacy helps and when it merely adds a cost. If size offers no protection or the vanguard already reaches the threshold, the contract removes no exposure disadvantage to offset its signing cost. Corollary~\ref{cor:common_rank} asks how much of the calculation depends on giving every rank equal weight. Uniform $H$ gives $p_H=1-w_\tau$ and exactly reproduces $G^{acct}$, while other weights retain the same benefit-minus-cost structure when both institutions use them. Those alternative weights need not describe equilibrium beliefs; Lemma~\ref{lem:global_selection} supplies the reason to use uniform ranks for the selection result.

\subsection{Designing the Contract: Thresholds and Failure Privacy}\label{subsec:design}

The architect must choose when to publish. A safe rule first needs a minimum size at which a signer is willing to be named, and it should then preserve the incentive to join as more people participate. Fix the marginal type $(e,\alpha)$ and audience-belief index $\mu$, with $\ell$ strictly decreasing in coalition mass, and define the type's \emph{protection requirement}
\begin{equation}
\rho(e,\alpha,\mu)=\inf\{y\in[0,1]:\ e\ge\alpha\ell(y,\mu)\},
\end{equation}
the smallest public coalition in which the type is willing to be named. If no size suffices, set $\rho=+\infty$. This is the homogeneous counterpart of the individual protection requirements $\rho_i$ in the companion paper.

A safe rule never names a signer below the protection requirement, and safety is evaluated after final coalition size is known. The companion paper extends this standard to composition-dependent protection through regret-free safety \citep{cashmanSafeCollectiveExpression2026}. A \emph{pointwise maximal} safe rule releases the list at every safe size.

\begin{proposition}[Maximal safe release]\label{prop:safe_disclosure}
Restrict attention to deterministic size-only publication indicators $\chi:[0,1]\to\{0,1\}$. The induced release map is
\begin{equation}
D_\chi(V)=
\begin{cases}
V, & \chi(m(V))=1,\\
\varnothing, & \chi(m(V))=0.
\end{cases}
\end{equation}
Thus the mechanism jointly publishes the complete signer set or no names. The indicator is \emph{ex post safe} if $\chi(y)=1\Rightarrow y\ge\rho$.
\begin{enumerate}
\item Within this class, the unique pointwise maximal safe indicator is
\begin{equation}
\chi^*(y)=\mathbf 1\{y\ge\rho\}.
\end{equation}
Every safe size-only rule therefore withholds the full list below $\rho$, and the maximal one releases it at every safe mass.
\item If the architect commits before anyone signs to withhold release below a chosen mass $\tau$ with finite $\rho\le\tau\le1$, the unique pointwise maximal safe indicator subject to that commitment is $\chi_\tau(y)=\mathbf 1\{y\ge\tau\}$. Below $\rho$ withholding is forced by safety; on $[\rho,\tau)$ it is chosen by the stronger commitment.
\end{enumerate}
\end{proposition}

Below $\rho$, safety forces the architect to withhold names. At or above that size, publication is safe, so a rule that releases as often as safety permits must publish. This conclusion applies to the full-list-or-nothing rules in the proposition. Choosing some names but not others, using identities to assess protection, or randomizing release would change the design problem.

Choosing a threshold above the minimum safe size buys extra protection at a cost, as Theorem~\ref{thm:selection_gap} shows. Raising $\tau$ keeps more candidate coalitions private and lowers $\ell(\tau,\mu)$, the exposure when publication first becomes possible. But it also raises $w_\tau$, the fraction needed to publish, and $k/(1-w_\tau)$, the cost that successful expression must cover. A higher threshold is harder to reach and requires a larger vanguard in favorable states. The architect must weigh stronger protection on success against the need to recruit more people.

Information alone cannot substitute for coalition mass when exposure dominates. Call the marginal type \emph{exposure-dominant} if $e<\alpha\ell(b(\theta),\mu')$ for every audience-belief index $\mu'$ attainable by a belief-only intervention, one that transmits information but releases no named coalition, so that an agent acting alone is publicly associated only with the vanguard mass $b(\theta)$. For such a type, no belief-only intervention makes vanguard-level public expression ex post safe: at any induced index $\mu'$, the lone expresser's payoff is $e-\alpha\ell(b(\theta),\mu')<0$, and under the maintained payoff technology safe association requires a coalition of mass at least $\rho(e,\alpha,\mu')>b(\theta)$. There are two routes to safe expression: a favorable enough audience-belief index or a larger public coalition.

Preserving the incentive to join is a further design requirement, beyond safety. Publishing a small coalition but withholding a larger one can make joining less attractive as participation rises, because extra endorsements could switch off a publication that a signer wanted. The next result rules out such reversals, and its richness condition ensures that some recruited type values publication near every size where the rule publishes. The \emph{closure} of the publication set contains those sizes and their boundary points.

\begin{lemma}[Threshold rules preserve complementarity]\label{lem:threshold_shape}
Consider deterministic size-only, failure-private publication indicators $\chi:[0,1]\to\{0,1\}$, with participation payoff $\chi(y)\bigl(e-\alpha\ell(y,\mu)\bigr)-k$ and nonpublication paying $-k$. Say the recruited class is \emph{rich} if for every mass $y$ in the \emph{closure} of the publication set there is a recruited type with $e>\alpha\ell(y,\mu)$. Under richness, the participation payoff is nondecreasing in the aggregate for every recruited type whose success payoff is nonnegative at every published mass if and only if the publication set is an upper set --- once the mechanism publishes at one mass, it also publishes at every larger mass --- and hence $\chi$ is a threshold indicator, up to the boundary convention at the cutoff.
\end{lemma}

Thresholds therefore serve two purposes. They keep names private below the chosen mass, and they ensure that additional participation never turns publication off or reverses the incentive to join.

The threshold helps through the privacy it provides on failure. Suppose instead that every signer's identity is released whether or not the threshold is reached, and that failed release gives the same expressive value and exposure cost as open expression at the realized coalition mass. The platform still collects signatures and checks a threshold, but failure no longer protects anyone. Under the baseline normalization in which only the contract incurs signing friction, the contract payoff at every state and rank is
\begin{equation}
u_{SAC}^{FD}(\theta,w)
=e-\alpha\ell\bigl(y(w;\theta),\mu\bigr)-k
=u_{open}(\theta,w)-k,
\end{equation}
so its rank-indifference value is $\bar e_{SAC}^{FD}(\theta)=\bar e_{open}(\theta)+k$. Complete failed-list disclosure eliminates the exposure-tail advantage and leaves the contract worse by exactly $k$ in expressive-benefit units; if the same friction is paid in both arms, the two rank-indifference values coincide.

\subsection{Trust and Recantation}\label{subsec:trust_recantation}

Signers must trust the contract, and the architect must allow for endorsements that may not last. A person deciding whether to sign may fear that the apparent contract is a scheme to collect names for the opposition. Even a genuine contract can collect valid endorsements from infiltrators who plan to repudiate the statement as soon as it appears. The first threat lowers the expected payoff from signing. The second requires a larger initial coalition so that enough people remain committed after publication.

At the interior evaluation state, write $q_\tau=1-w_\tau$ for the share of ranks at which an honest contract publishes, and let
\begin{equation}
h_b(\theta)=\alpha\ell\bigl(b(\theta),\mu\bigr)
\end{equation}
be the exposure harm from being identified without protection beyond the vanguard. Suppose the administrator is honest with probability $1-\delta_F$ and false with probability $\delta_F\in[0,1)$. A false administrator leaks every signer's name to the opposition, creates no protecting public coalition, and delivers no expressive benefit. The expected payoff from signing is then
\begin{equation}
u_{SAC}^{F}(\theta,w)
=(1-\delta_F)\ind\{y(w;\theta)\ge\tau\}
\bigl(e-\alpha\ell(y(w;\theta),\mu)\bigr)
-\delta_F h_b(\theta)-k.
\label{eq:false_contract_payoff}
\end{equation}
The probability $1-\delta_F$ is the signer's trust that the advertised contract is genuine.

\begin{proposition}[Trust thresholds and recantation margins]\label{prop:trust_recantation}
The two threats have the following effects.
\begin{enumerate}
\item \emph{False-contract risk.} The contract's rank-indifference value becomes
\begin{equation}
\bar e_{SAC}^{F}(\theta)
=\bar e_{SAC}(\theta)
+\frac{\delta_F\bigl(h_b(\theta)+k\bigr)}{(1-\delta_F)q_\tau}.
\label{eq:false_contract_rank_value}
\end{equation}
It equals the baseline value at $\delta_F=0$ and rises strictly with $\delta_F$. Under the false-contract selection domain in Supplemental Appendix~\xref{app:trust_recantation}, Corollary~\xref{cor:false_contract_selection} extends the noise-limit selection result to equation~\eqref{eq:false_contract_payoff}. Under that domain, the selected expressive-benefit requirement is
\begin{equation}
e_{SAC}^{sel,F}(\theta)
=\max\{\bar e_{SAC}^{F}(\theta),\alpha\ell(\tau,\mu)\}.
\end{equation}
For a type that satisfies $\tau$-protection and has $e>\bar e_{SAC}(\theta)$, define the minimum trust needed to pass the rank-indifference test by
\begin{equation}
t_{sign}^{*}(\theta;e)
=\frac{h_b(\theta)+k}
{h_b(\theta)+k+q_\tau\bigl(e-\bar e_{SAC}(\theta)\bigr)}.
\label{eq:signer_trust_threshold}
\end{equation}
In the noise limit, this type lies on the participating side of the rank boundary when $1-\delta_F>t_{sign}^{*}$ and on the nonparticipating side when $1-\delta_F<t_{sign}^{*}$, with rank indifference at equality. Holding the other terms in equation~\eqref{eq:signer_trust_threshold} fixed, minimum trust rises with leak harm $h_b(\theta)$.

This participation threshold is distinct from the trust needed for the contract to outperform open expression. If $G^{acct}(\theta)>0$ and $\bar e_{open}(\theta)>\alpha\ell(\tau,\mu)$, the selected gap remains positive exactly when
\begin{equation}
1-\delta_F>
t_{adv}^{*}(\theta)
\equiv
\frac{h_b(\theta)+k}
{h_b(\theta)+k+q_\tau G^{acct}(\theta)}.
\label{eq:advantage_trust_threshold}
\end{equation}
Holding the other terms fixed, this minimum trust also rises with leak harm.

\item \emph{Bounded recantation.} Consider a nonatomic post-release extension in which authenticated endorsements may come from infiltrators as well as supporters. Recantations are public before the opposition responds, and a signer's exposure depends on the mass $m(C^{stay})$ that continues to stand behind the statement, not on the larger historical list first released. Let $0<\rho\le1$ be the minimum safe continuing mass. Suppose every mass up to the maximum authenticated-endorsement capacity is attainable and the worst-case uncertainty set permits any recanting subset allowed by the bounds below. If at most mass $\bar z\ge0$ can recant, the smallest size threshold that guarantees $m(C^{stay})\ge\rho$ is
\begin{equation}
\tau_{\bar z}=\rho+\bar z.
\label{eq:absolute_recantation_margin}
\end{equation}
Nontrivial robust release requires this threshold not to exceed the maximum authenticated-endorsement capacity. Alternatively, suppose a known vanguard of fixed mass $b_0<\rho$ remains committed and at most a fraction $\phi\in[0,1)$ of endorsements beyond that vanguard can recant. The smallest robust threshold is
\begin{equation}
\tau_\phi
=b_0+\frac{\rho-b_0}{1-\phi},
\qquad
\tau_\phi-\rho
=\frac{\phi(\rho-b_0)}{1-\phi}.
\label{eq:proportional_recantation_margin}
\end{equation}
Again, nontrivial robust release requires enough authenticated-endorsement capacity to reach $\tau_\phi$. Both margins are sharp within this uncertainty set: any smaller attainable threshold permits an allowed recantation that leaves less than $\rho$ publicly committed.
\end{enumerate}
\end{proposition}

Supplemental Appendix~\xref{app:trust_recantation} proves the proposition and states the false-contract selection conditions. The model takes the threats as given, treating $\delta_F$ as the signer's assessed probability of a false contract and $\bar z$ and $\phi$ as credible bounds on recantation. It does not derive these quantities from the strategic choices of administrators or infiltrators. Authentication can establish eligibility and prevent duplicate endorsements, but it cannot establish motive.

The recantation result assumes that only people who continue to stand behind the statement protect one another. If the opposition instead targets the original named list, even repudiated endorsements can occupy its attention, so no extra margin for recantation is needed. The result answers how many names a safe release requires, not how recantation risk changes the decision to sign. That second question cannot be answered by simply replacing $\tau$ in Theorem~\ref{thm:selection_gap}. When protection depends on $C^{stay}$, exposure on success depends on the mass that remains committed. A participation analysis must change that exposure term and recheck complementarity, dominance, and the rank crossing.

A working design must remove enough exposure to justify signing, set a protective and reachable threshold, convince people that their names are in trustworthy custody, and allow for any credible bound on recantation.

\section{Applications}\label{sec:applications}

The model points to settings where an identifiable group faces retaliation and a larger public coalition protects each member. The group needs a committed vanguard, a protective threshold it can reach, and an administrator or protocol trusted to keep failed lists private.

Applying the model requires knowing who can reach, trust, punish, and observe whom. Recruitment ties determine who hears the invitation and help set vanguard mass $b(\theta)$ and participation capacity $Y(\theta)$. Trust in the architect and administrator determines whether people believe the invitation and publication condition, and distrust shrinks the effective vanguard or capacity and raises signing friction $k$. The employer, peers, professional body, or public able to punish signers determine $\mu$ and $\ell$. How that audience learns about the campaign determines whether it sees a failed attempt despite seeing no names. Private recruitment can preserve attempt privacy; public advertising or a visible counter can reveal the attempt.

Safety in numbers has a concrete basis when retaliation consumes scarce attention, managerial time, investigative resources, or political capital. Punishing a larger group is also more visible and costly, and its members can document, contest, and publicize retaliation together. In these settings, a larger released coalition lowers each signer's material cost ($\ell(y',\mu)<\ell(y,\mu)$ for $y'>y$).

Whether a larger group supplies this protection is an empirical question. If the audience can sanction each additional signer just as easily, exposure stays flat and the contract has no exposure-tail advantage. If a larger coalition provokes stronger sanctions, exposure can rise instead. A different model is also needed when one report suffices or only the first reporter receives a reward. Others' participation can then discourage one's own, reversing the complementarity on which the comparison rests.

Collective whistleblowing is one possible application. Employees in a defined unit form the eligible population, and an employer or professional audience can retaliate against them. Wells Fargo's independent investigation documents sales pressure, internal complaints, and employee terminations \citep{independentdirectorsoftheboardofwellsfargoandcompanySalesPracticesInvestigation2017}. CFPB and DOJ records document the enforcement that followed \citep{u.s.consumerfinancialprotectionbureauConsumerFinancialProtection2016,u.s.departmentofjusticeWellsFargoAgrees2020}. These records establish the retaliation problem, but applying the model would also require evidence that a larger group protects reporters and that a committed vanguard can clear the threshold in favorable states. Protected or already-disclosed reporters could supply that vanguard. A regulator, independent counsel, or group of trustees could verify employee status, hold submissions, and release a named collective report at the threshold, keeping names private below it. Each reporter would pay signing friction $k$ and be named on success.

Pre-petition labor organizing also fits when known organizers or openly pro-union workers supply a credible vanguard and a larger committee lowers each member's exposure. The prospective bargaining unit is the eligible population, the employer is the audience, and the threshold is the target committee size before members go public. This stage precedes a certification election, in which the NLRB converts a majority vote into certification and formal representation \citep{u.s.nlrbConductElectionsNational2026,u.s.nlrbYourRightForm2026}. Conditional disclosure then lets organizers recruit a named committee without exposing workers if the committee stays too small to launch.

Public letters and threshold commitments use the same structure. Senior or already-public signers form the vanguard, and an independent committee or platform holds the list. The model fits when the vanguard is credible, the threshold is reachable, and appearing on a short list is especially costly, with that cost falling as the list grows. The same conditions apply to public statements backed by a fixed constituency.

Custody makes the publication promise credible, and an independent holder can enforce the rule directly. A distributed design instead combines authenticated encrypted submissions, a verifiable threshold decision, and a required number of independent trustees acting together. Under the design's security assumptions, no single administrator can read or release a failed list. Secure allegation escrows perform closely related functions \citep{arunFindingSafetyNumbers2020}, and \citet{mangipudiCollusionDeterrentThresholdInformation2023} develop a collusion-deterrent threshold escrow. Supplemental Appendix~\xref{app:implementation} describes how to combine these components and which security and governance assumptions the design needs. The proposed architecture still requires implementation and audit. Counsel, a regulator, a union, or a platform can operate the escrow without joining the coalition. Proposition~\ref{prop:trust_recantation} shows how much doubt signers can tolerate about custody before they prefer to stay out.

\subsection{Testable Implications} Experiments can separate the value of conditional disclosure from the costs of arranging it. Under the selection conditions, the contract lowers the participation cutoff when avoided exposure outweighs signing costs adjusted for publication risk and the threshold protects those recruited. The vanguard and capacity conditions matter: the selection argument needs a vanguard that clears the threshold in favorable states, and publication requires a threshold within reach. Raising the threshold improves protection on publication but demands more endorsements and makes successful expression bear a larger share of signing costs. Complete failed-list disclosure eliminates the exposure-tail advantage.

Trust and recantation lead to two further tests. Greater doubt that the contract is genuine raises the value of expression needed to sign, especially when a leak would be harmful. A larger credible bound on recantation requires more names at publication and can put safe release beyond reach. Laboratory and field experiments can vary vanguard size, threshold, custody, failure disclosure, administrator credibility, and recantation risk separately to identify these effects.

These tests concern participation, not whether the resulting coalition benefits society. Social value depends on the statement, the action it produces, and its effects on nonsigners and targets.

\section{Concluding Remarks}\label{sec:discussion}
Conditional disclosure lets people agree to speak together without being named if too few others join. In the model, this removes exposure between the vanguard and the publication threshold. The exposure-tail decomposition measures that benefit against signing costs adjusted for publication risk, holding the state and the audience's beliefs fixed. Turning the payoff advantage into additional participation also requires a threshold that protects the people recruited. The contract expands participation when both tests are met, the vanguard clears the threshold in favorable states, and failed lists remain private.

A survey reveals hidden support; a contract helps turn that support into a named coalition. Compared with open expression, the contract changes what happens to someone who agrees to participate and finds that few others do: under open expression her name is public, while under the contract it stays private. Publishing failed lists removes this advantage and leaves only the extra cost of signing.

The vanguard makes joining worthwhile in favorable states even without other marginal participants. The threshold sets the protection offered to new signers, and failure privacy delivers it. Doubt about custody raises the value of expression needed to sign, and a bound on recantation determines how many extra names a safe release needs. The Supplemental Appendix shows how the decision to sign changes in a finite group, where one endorsement can trigger publication, and describes how existing institutional and cryptographic components can support the contract. Deployment requires those components, their governance, and the user interface to withstand the relevant attacks.

The exposure-tail formula predicts which contracts expand public coalitions and which merely add costs.

\let\articleSection\section
\appendix

\section{Appendix A. Proofs}\label{app:proofs}

Most results follow from payoff accounting or from the fact that exposure falls with coalition size; the proof of Lemma~\ref{lem:global_selection} is longer. The proofs below follow the order of the results in the text.

\begin{proof}[Proof of Proposition~\getrefnumber{prop:no_selection}]
Because $\ell$ strictly decreases with coalition mass, expression becomes more attractive as more people express. In the vanguard-only profile, a marginal expresser faces mass $b(\theta)$ and payoff $e-\alpha\ell(b(\theta),\mu)<0$, so silence is a strict best response. In the full-expression profile, she faces $Y(\theta)$ and payoff $e-\alpha\ell(Y(\theta),\mu)\ge0$, so expression is a best response. Both profiles are therefore equilibria. Information matters only by changing $\mu$ or the evaluated state; multiplicity remains at every fixed pair inside the displayed band.
\end{proof}

\begin{proof}[Proof of Lemma~\getrefnumber{lem:global_selection}]
First, better states and greater participation make joining more attractive, which identifies signals at which one action is dominant. Next, two cutoff sequences bound the surviving strategies from opposite sides. Finally, an agent at the cutoff evaluates ranks uniformly as signal noise vanishes. Her payoff therefore converges to the rank average $R$.

Fix one arm and omit its subscript, writing $u$, $R$, and $\theta^*$ for its payoff, rank integral, and unique crossing.

\textit{Pointwise deletion.} The process covers every measurable strategy, including asymmetric and nonmonotone profiles, and evaluates each agent at every signal, so the conclusion holds at every signal rather than almost everywhere. In the diffuse branch, for every measurable state set $B$, the posterior kernel---the probability assigned to the set after observing the signal---is
\[
P_\sigma(B\mid x)=\int \mathbf 1\{x-\sigma\eta\in B\}\,dF_\eta(\eta).
\]
In the proper-prior branch, the posterior kernel is
\[
P_\sigma(d\theta\mid x)
=\frac{g(\theta)f_\eta((x-\theta)/\sigma)\,d\theta}
{\int g(t)f_\eta((x-t)/\sigma)\,dt}.
\]
Write $P_\sigma(\cdot\mid x)$ for the applicable branch. For every agent $i$ and signal $x$, initially both actions are available: $\mathcal{A}_i^0(x)=\{0,1\}$. Given the action sets at round $n$, let $\mathcal S^n$ contain all jointly measurable profiles $\varphi$ satisfying $\varphi_i(x)\in\mathcal A_i^n(x)$ for every $(i,x)$. Against any surviving profile $\varphi\in\mathcal S^n$, the \emph{interim} payoff gain from participating is
\begin{equation}
\Psi_i(x;\varphi)=\int u\bigl(\theta,w_\varphi(\theta)\bigr)P_\sigma(d\theta\mid x).
\end{equation}
Deleted actions stay deleted. At round $n+1$, action 1 is deleted from $\mathcal A_i^n(x)$ when $\sup_{\varphi\in\mathcal S^n}\Psi_i(x;\varphi)<0$: joining loses even under the most favorable remaining behavior by others. Action 0 is deleted when $\inf_{\varphi\in\mathcal S^n}\Psi_i(x;\varphi)>0$: joining wins even under the least favorable remaining behavior. The surviving strategies are the profiles in $\mathcal S^\infty=\bigcap_n\mathcal S^n$.

\textit{Dominance regions.} In the open arm, $y(w;\theta)$ rises with the state and marginal participation, while $\ell$ falls with coalition mass. In the contract arm, payoff is $-k$ below the threshold, jumps weakly upward at publication by $\tau$-protection, and then rises. Both payoffs are therefore bounded and nondecreasing in the state and in others' participation. Because $Y$ is weakly increasing, every state $s\le\theta_L$ has $Y(s)\le Y(\theta_L)<\tau$, so the contract fails at every rank and pays $-k$, while open expression pays at most $e-\alpha\ell(Y(\theta_L),\mu)<0$ by low-state exposure. Because $b$ is weakly increasing, every state $s\ge\theta_H$ has $b(s)\ge\tau$, so the contract publishes at every rank and both arms pay at least $e-\alpha\ell(b(\theta_H),\mu)-k>0$ by vanguard clearance and protectability. Hence there are constants $m>0$ and $M$ such that, at every rank, $u(s,\cdot)\le-m$ for $s\le\theta_L$, $u(s,\cdot)\ge m$ for $s\ge\theta_H$, and $|u|\le M$ everywhere.

Signal dominance follows once noise is small. Fix $\delta>0$ and a signal $x\le\theta_L-\delta$. In the diffuse branch, the represented state $x-\sigma\eta$ lies at or below $\theta_L$ whenever $\eta\ge-\delta/\sigma$, so against every admissible profile
\[
\Psi_i(x;\varphi)\le-m\bigl[1-F_\eta(-\delta/\sigma)\bigr]+M\,F_\eta(-\delta/\sigma),
\]
which is negative once $\sigma$ is small, uniformly over such signals and profiles. In the proper-prior branch, let $\bar g$ bound $g$ and let $\underline g>0$ bound it below on $[\theta_L-3\delta/2,\theta_L-\delta/2]$, which lies where $g$ is positive once $\delta$ is small. The posterior probability of states above $\theta_L$ at the signal $\theta_L-\delta$ is at most $\bar g\,F_\eta(-\delta/\sigma)$ divided by $\underline g\,\bigl[F_\eta\bigl(\tfrac{\delta}{2\sigma}\bigr)-F_\eta\bigl(-\tfrac{\delta}{2\sigma}\bigr)\bigr]$, which tends to zero as $\sigma\downarrow0$; monotone likelihood ratios make that probability nondecreasing in the signal, so the same bound holds at every $x\le\theta_L-\delta$, and the payoff bound above applies with that probability in place of $F_\eta(-\delta/\sigma)$. The high side is symmetric. Thus for small $\sigma$, staying out is a strict best response at every signal $x\le x_L\equiv\theta_L-\delta$ and joining at every $x\ge x_H\equiv\theta_H+\delta$, whatever strategies others use. Let $\mathcal{J}=[x_L,x_H]$, with $\delta$ small enough in the proper-prior branch that $g$ is positive on a neighborhood of $\mathcal{J}$.

\textit{Cutoff bounds for pointwise deletion.} The gain $\Psi_i(x;\varphi)$ tells us which actions can be deleted against arbitrary profiles in $\mathcal S^n$. We can bound the surviving choices with a simpler calculation: suppose everyone else joins above the same signal cutoff $c$. Define
\begin{equation}
\Psi_\sigma(x;c)
=\E\!\left[u\bigl(\theta,w_c(\theta)\bigr)\mid x_i=x\right],
\qquad
w_c(\theta)=1-F_\eta\!\left(\frac{c-\theta}{\sigma}\right),
\end{equation}
with $w_{-\infty}=1$ and $w_{+\infty}=0$. The function $\Psi_\sigma(x;c)$ is the expected gain from joining after signal $x$ when all other marginal agents use cutoff $c$. A higher signal makes favorable states more likely, so monotone likelihood ratios and state-monotone payoffs make $\Psi_\sigma$ nondecreasing in $x$. A higher cutoff reduces others' participation, so strategic complementarity makes $\Psi_\sigma$ nonincreasing in $c$.

The posterior kernels give an explicit continuity argument. In the diffuse branch,
\begin{equation}
\Psi_\sigma(x;c)=
\int u\!\left(x-\sigma\eta,
1-F_\eta\!\left(\eta+\frac{c-x}{\sigma}\right)\right)dF_\eta(\eta).
\end{equation}
In the proper-prior branch, the same expression is
\begin{equation}
\frac{\int u\!\left(x-\sigma\eta,
1-F_\eta\!\left(\eta+\frac{c-x}{\sigma}\right)\right)
g(x-\sigma\eta)f_\eta(\eta)\,d\eta}
{\int g(x-\sigma\eta)f_\eta(\eta)\,d\eta}.
\end{equation}
The denominator is positive on $\mathcal{J}$. For any convergent sequence $(x_j,c_j)$, the integrands converge at every error realization except possibly where limiting coalition mass equals $\tau$. As $\eta$ rises, both the state and others' participation fall, so equality holds for at most one value of $\eta$. The continuous error distribution assigns that value probability zero. Bounded payoffs and continuous primitives then allow dominated convergence to pass the limit through the integral. Thus $\Psi_\sigma$ is jointly continuous in the signal and conjectured cutoff in either branch. Extending $g$ by zero creates possible discontinuities only at the endpoints of its support, which also have zero probability under the continuous error density.

The cutoff $-\infty$ is the optimistic conjecture that every marginal agent participates; $+\infty$ is the pessimistic conjecture that none does. Start with $\underline c_0=-\infty$ and $\bar c_0=+\infty$. The new lower cutoff is the first signal at which joining pays under the previous optimistic bound. The new upper cutoff is the last signal at which joining does not pay under the previous pessimistic bound:
\begin{equation}
\begin{aligned}
\underline c_{n+1}
&=\inf\{x:\Psi_\sigma(x;\underline c_n)\ge0\},\\
\bar c_{n+1}
&=\sup\{x:\Psi_\sigma(x;\bar c_n)\le0\}.
\end{aligned}
\end{equation}
Signal dominance places both first-round boundaries in $\mathcal{J}$, since $\Psi_\sigma(x;c)<0$ for $x\le x_L$ and $\Psi_\sigma(x;c)>0$ for $x\ge x_H$ at every $c$, and monotonicity keeps every later boundary there. Each round moves the optimistic lower boundary weakly upward and the pessimistic upper boundary weakly downward:
\begin{equation}
\underline c_n\uparrow\underline c_\sigma,
\qquad
\bar c_n\downarrow\bar c_\sigma,
\qquad
\underline c_\sigma\le\bar c_\sigma.
\end{equation}
Define $\beta^-(c)=\inf\{x:\Psi_\sigma(x;c)\ge0\}$ and $\beta^+(c)=\sup\{x:\Psi_\sigma(x;c)\le0\}$. These are the lower and upper best-response boundaries to cutoff $c$. Both maps are nondecreasing, and continuity gives $\beta^-(c)\le\beta^+(c)$. Because $\underline c_{n+1}=\beta^-(\underline c_n)$ and $\bar c_{n+1}=\beta^+(\bar c_n)$, induction from $-\infty\le+\infty$ gives the two monotone directions and
\begin{equation}
\underline c_{n+1}
=\beta^-(\underline c_n)
\le\beta^-(\bar c_n)
\le\beta^+(\bar c_n)
=\bar c_{n+1}.
\end{equation}

The two cutoff sequences bound every surviving strategy, even though those strategies need not themselves be cutoffs. After $n$ rounds, every survivor obeys
\begin{equation}
\mathbf 1\{x>\bar c_n\}
\le \varphi_i(x)\le
\mathbf 1\{x\ge\underline c_n\}.
\end{equation}
Every surviving strategy prescribes participation above the upper boundary and nonparticipation below the lower boundary; only the interval between them remains unresolved. The base case is $0\le \varphi_i(x)\le1$. If the bounds hold at round $n$, exact aggregation and continuity of $F_\eta$ imply for every measurable surviving profile $\varphi$ that
\begin{equation}
w_{\bar c_n}(\theta)\le w_\varphi(\theta)\le w_{\underline c_n}(\theta)
\qquad\text{for every }\theta.
\end{equation}
If $x>\bar c_{n+1}$, then $\Psi_\sigma(x;\bar c_n)>0$. Joining pays even under the lowest participation allowed by the previous bounds, so action 0 is deleted against every profile in $\mathcal S^n$. If $x<\underline c_{n+1}$, then $\Psi_\sigma(x;\underline c_n)<0$. Joining loses even under the highest allowed participation, so action 1 is deleted. This proves the next-round bounds. The bounding cutoffs are comparison devices, not necessarily surviving profiles.

\textit{Limits of the cutoff iterations.}

Continuity and dominance make the agent at each finite-round boundary indifferent:
\begin{equation}
\Psi_\sigma(\underline c_{n+1};\underline c_n)=0,
\qquad
\Psi_\sigma(\bar c_{n+1};\bar c_n)=0.
\end{equation}
Passing to the limits yields
\begin{equation}
\Psi_\sigma(\underline c_\sigma;\underline c_\sigma)=0,
\qquad
\Psi_\sigma(\bar c_\sigma;\bar c_\sigma)=0.
\end{equation}
Hence $\mathbf 1\{x\ge\underline c_\sigma\}$ and $\mathbf 1\{x>\bar c_\sigma\}$ are best responses to themselves under opposite tie rules at exact indifference. Both endpoint profiles survive every round of deletion. Each is initially admissible, and if it survives through round $n$ it belongs to $\mathcal S^n$; each of its prescribed actions is a best response to it, so neither action is strictly dominated against every profile in $\mathcal S^n$, and both survive round $n+1$. Both endpoints therefore belong to $\mathcal S^\infty$, and the finite-round bounds are tight: the lower and upper boundaries of the surviving strategies are exactly $\underline c_\sigma$ and $\bar c_\sigma$. Continuing deletion beyond the finite rounds, as some definitions of rationalizability do, removes nothing further, because the same best-response argument applies at every later stage.

\textit{Rank evaluation and the noise limit.} Under the diffuse posterior kernel, an agent at the boundary $x_i=c$ represents the unknown state as $\theta=c-\sigma\eta$, with $\eta\sim F_\eta$. The probability-integral transform says that applying a continuous distribution function to a draw from that distribution produces a uniform random variable. It therefore makes $1-F_\eta(\eta)$ uniform on $[0,1]$, so the boundary payoff is
\begin{equation}
\Psi_\sigma(c;c)
=\int_0^1u\bigl(c-\sigma F_\eta^{-1}(1-w),w\bigr)\,dw.
\end{equation}
The uniform weighting thus comes from the private signals. An agent at the boundary remains uncertain about what fraction of others received higher signals, even as her uncertainty about the state shrinks. The calculation assigns equal weight to every fraction between none and all. For any $c_\sigma\to c\in\mathcal{J}$, bounded payoffs and dominated convergence give $\Psi_\sigma(c_\sigma;c_\sigma)\to R(c)$. The contract payoff jumps at the threshold, but the jump occurs at no more than one rank and does not affect the integral.

Under a proper prior, extend $g$ by zero outside its compact support. The posterior density of $\eta=(c_\sigma-\theta)/\sigma$ is proportional to $g(c_\sigma-\sigma\eta)f_\eta(\eta)$, and the same payoff is
\begin{equation}
\frac{\int u\bigl(c_\sigma-\sigma\eta,1-F_\eta(\eta)\bigr)
g(c_\sigma-\sigma\eta)f_\eta(\eta)\,d\eta}
{\int g(c_\sigma-\sigma\eta)f_\eta(\eta)\,d\eta}.
\end{equation}
The denominator converges to $g(c)>0$, and dominated convergence again yields $R(c)$. As in the diffuse case, the single moving publication rank has posterior probability zero and does not change the limit.

Both boundary sequences lie in the compact interval $\mathcal{J}$, so every subsequence has a convergent further subsequence. At any limit $c$, boundary indifference and either rank calculation give $R(c)=0$. Strict single crossing gives the unique zero $\theta^*$. Every convergent subsequence therefore has the same limit, so both full sequences converge to $\theta^*$:
\begin{equation}
\underline c_\sigma\longrightarrow\theta^*,
\qquad
\bar c_\sigma\longrightarrow\theta^*.
\end{equation}
Every surviving strategy lies between these cutoffs. For every $\delta>0$, sufficiently small $\sigma$ therefore forces participation above $\theta^*+\delta$ and nonparticipation below $\theta^*-\delta$.
\end{proof}

\begin{proof}[Proof of Theorem~\getrefnumber{thm:selection_gap}]
Under open expression, expressive value and exposure occur at every rank, so the rank integral is $e-\alpha\int_0^1\ell(y(w;\theta),\mu)\,dw$. Setting it to zero gives the first rank-indifference value. Under the contract, expressive value and exposure occur only at or above $w_\tau$, while friction is paid at every rank:
\begin{equation}
\int_0^1 u_{SAC}(\theta,w)\,dw=(1-w_\tau)\,e-\alpha\int_{w_\tau}^1\ell\bigl(y(w;\theta),\mu\bigr)dw-k.
\end{equation}
Setting this expression to zero gives the second value. For $G^{acct}$, subtract the contract value from the open value, split the open integral below and above the threshold, and write each part as interval length times its average. In the monotone contract domain, both rank indifference and $\tau$-protection must hold. The smallest admissible expressive value is therefore $e^{sel}_{SAC}=\max\{\bar e_{SAC},\alpha\ell(\tau,\mu)\}$. Subtracting this maximum from $\bar e_{open}$ and applying $x-\max\{y,z\}=\min\{x-y,x-z\}$ yields $G^{sel}$.

With symmetric friction, subtracting $k$ from the open payoff adds $k$ to its expressive-benefit cutoff. The contract cutoff is unchanged. Adding $k$ to the baseline gap changes its friction term from $-k/(1-w_\tau)$ to $-k w_\tau/(1-w_\tau)$.
\end{proof}

\begin{proof}[Proof of Corollary~\getrefnumber{cor:cutoff_states}]
At $\theta^*_{open}$ the open-expression rank integral is zero, so $e=\bar e_{open}(\theta^*_{open})$. At an interior state the contract's rank integral equals $(1-w_\tau)\bigl(e-\bar e_{SAC}(\theta)\bigr)$, so at $\theta^*_{open}$ it has the sign of $\bar e_{open}(\theta^*_{open})-\bar e_{SAC}(\theta^*_{open})=G^{acct}(\theta^*_{open})$. By strict single crossing, $R_{SAC}$ is negative below its unique zero and positive above it, so $\theta^*_{SAC}<\theta^*_{open}$ exactly when $R_{SAC}(\theta^*_{open})>0$.
\end{proof}

\begin{proof}[Proof of Corollary~\getrefnumber{cor:gap_properties}]
For part 1, let each $(e,\alpha)$ index a separate economy. If it satisfies Assumption~\ref{ass:operative_domain}, Lemma~\ref{lem:global_selection} applies to both arms. The inequalities place its marginal type above the contract cutoff $e^{sel}_{SAC}$ and below the open-expression cutoff $\bar e_{open}$. The lower-bound maximum imposes rank indifference and $\tau$-protection. The interval is nonempty exactly when $G^{sel}>0$; equality at the recruitment-coverage boundary leaves no strict interval. For part 2, without safety in numbers, exposure averages below and above the threshold are equal, so their difference vanishes. For part 3, as the vanguard approaches the threshold, $w_\tau\downarrow0$ and the protected exposure tail vanishes, leaving signing cost $k$. Part 4 compares the rank averages and threshold protection at the same fixed $\mu$; neither a type distribution nor a different audience index enters.
\end{proof}

\begin{proof}[Proof of Corollary~\getrefnumber{cor:common_rank}]
Under common weighting $H$, open expression has expected rank payoff $e-\alpha\E_H[L(W)]$, which yields $\bar e_{open}^H$. Under the contract, expressive value and exposure occur only when $W\ge w_\tau$, with probability $p_H$, while friction is always paid. Its expected rank payoff is
\begin{equation}
p_H e-\alpha p_H\E_H[L(W)\mid W\ge w_\tau]-k,
\end{equation}
which yields $\bar e_{SAC}^H$. The law of total expectation gives
\begin{equation}
\E_H[L(W)]
=(1-p_H)\E_H[L(W)\mid W<w_\tau]
+p_H\E_H[L(W)\mid W\ge w_\tau].
\end{equation}
Substitution and subtraction give the identity. Because $y(w;\theta)$ rises in $w$ and $\ell$ strictly decreases in coalition mass, exposure at every below-threshold rank exceeds exposure at every successful rank. Since $p_H\in(0,1)$, both regions have positive probability, so their conditional exposure averages differ strictly.
\end{proof}

\begin{proof}[Proof of Proposition~\getrefnumber{prop:safe_disclosure}]
For part 1, if $\rho=+\infty$, no coalition size protects the type. Safety forces $\chi(y)=0$ on $[0,1]$, the unique maximal safe indicator. Otherwise, every admissible binary indicator must withhold below $\rho$. At $y\ge\rho$, jointly publishing the complete signer set is safe, so $\chi^*(y)=1$. Thus $\chi^*$ publishes wherever any safe indicator can and weakly dominates every safe indicator pointwise. Any other safe indicator withholds at some $y\ge\rho$ and is strictly less permissive there. The stated class permits no partial-list release.

For part 2, the commitment forces $\chi(y)=0$ below $\tau$. Because $\tau\ge\rho$, every coalition at or above the threshold protects the type, so joint publication is safe. The indicator $\chi_\tau(y)=\mathbf 1\{y\ge\tau\}$ is feasible and publishes at every permitted size. It weakly dominates every other safe indicator satisfying the commitment, strictly wherever another indicator withholds at $y\ge\tau$.
\end{proof}

\begin{proof}[Proof of Lemma~\getrefnumber{lem:threshold_shape}]
First suppose $\chi=\mathbf{1}\{y\ge\tau\}$. For a type with nonnegative success payoff wherever publication occurs, participation pays $-k$ below the threshold, jumps weakly upward at the threshold, and rises above it because $\ell$ strictly decreases with coalition mass. Participation payoff is therefore nondecreasing in size.

Conversely, suppose the publication set is not an upper set. Then for some $y_1<y_2$, the mechanism publishes at $y_1$ but withholds at $y_2$. Let $\underline y=\inf\{y:\chi(y)=1\}$ be the lower edge of the publication set. Richness supplies a recruited type with $e>\alpha\ell(\underline y,\mu)$. Since exposure falls with size and every published mass is at least $\underline y$, this type has positive success payoff whenever publication occurs. Yet its participation payoff falls from $e-\alpha\ell(y_1,\mu)-k>-k$ at $y_1$ to $-k$ at the larger $y_2$. This violates monotonicity for a recruited type.

Richness is necessary for the converse. Without it, an indicator can switch on irregularly at low sizes where publication benefits no recruited type. Those irrelevant publication points need not reverse any relevant payoff even though the publication set is not an upper set.
\end{proof}

\let\section\articleSection

\bibliographystyle{aea}
\bibliography{conditional_disclosure}

\ifbuildwithsupplement
  \clearpage
  \documentclass[conditional_disclosure.tex]{subfiles}

\begin{document}

\ifSubfilesClassLoaded{%
  \title{Supplemental Appendix\\Conditional Disclosure as a Coordination Device}
  \shortTitle{Supplemental Appendix: Conditional Disclosure}
  \issueName{}
  \ifanonymizeforpeerreview
    \author{Anonymous}
  \else
    \author{Matthew Cashman\thanks{Cashman: MIT Sloan School of Management, cashman@mit.edu.}}
  \fi
  \date{}
  \begin{abstract}
  \end{abstract}
  \maketitle
}{ }

\startsupplement

\ifSubfilesClassLoaded{ }{%
  \appendixparttitle{Supplemental Appendix}
}

This Supplemental Appendix collects material that the article's main comparison does not need: a finite model, under explicit regularity conditions, in which one endorsement can trigger publication (Appendix~\ref{app:finite_vanguard}); the distinction between failure privacy and attempt privacy, with an informational-exchangeability lemma (Appendix~\ref{app:privacy_exchangeability}); the trust threshold and recantation margin (Appendix~\ref{app:trust_recantation}); and an implementation architecture (Appendix~\ref{app:implementation}). The continuum selection proof and common-rank accounting result remain in the article.

\section{Appendix B. Finite Vanguard Benchmark}\label{app:finite_vanguard}

The article uses a continuum of marginal agents, so no single marginal agent can change the publication outcome. This section instead studies a fixed eligible population $N$ with $n$ marginal agents. Each person contributes one full endorsement and can be pivotal: her endorsement can move the count across the publication threshold. The question is whether a simple cutoff rule still describes behavior when each agent may determine publication but does not know how many others will sign.

Let $S\subseteq N$ be the supporters willing to be published when the resulting coalition protects them. At each state, the model tracks $n$ marginal agents and a vanguard. The state $\theta\in[\underline\theta,\overline\theta]$ measures how favorable the environment is to participation. Before receiving private information, agents share a prior density $g$ that is positive and continuous throughout the state space and describes which states the agents initially consider more likely. The integer publication threshold satisfies $K\ge2$. A weakly increasing measurable function
\begin{equation}
b:[\underline\theta,\overline\theta]\longrightarrow\{0,1,\ldots,K\}
\end{equation}
records the number of vanguard endorsements in each state, with $n+b(\theta)\le |S|$. Thus $b(\theta)$ is a count, not a probability. Vanguard members sign regardless of what the marginal agents do. As in the article's simultaneous benchmark, they need not be publicly visible during collection; their endorsements are committed before the marginal choices resolve. The companion paper's sequential benchmark instead defines the vanguard as support already public before recruitment. Because $b$ is weakly increasing, a more favorable state never has fewer vanguard signers.

Everyone knows $b$ and the prior, but no marginal agent observes the realized state or the resulting vanguard count $b(\theta)$ before acting. During collection, each marginal agent sees only her own signal, not other agents' signals or decisions or an interim endorsement count. The administrator holds all received endorsements as a private signer set $V\subseteq S$. The threshold release map is
\begin{equation}
D_K(V)=
\begin{cases}
V, & |V|\ge K,\\
\varnothing, & |V|<K.
\end{cases}
\end{equation}
On success, the released signer set becomes the public coalition $C=D_K(V)$.

Agent $i$ observes
\begin{equation}
x_i=\theta+\sigma\eta_i,
\end{equation}
where $\sigma>0$ controls how noisy the signal is. The shocks $\eta_i$ are independent across agents, independent of $\theta$, and identically distributed with continuous strictly increasing distribution $F_\eta$ and positive density $f_\eta$. The signal structure satisfies monotone likelihood ratios, so a higher signal makes higher states relatively more likely. If agent $i$ signs and $Q_{-i}$ other marginal agents sign, the publication condition is
\begin{equation}
b(\theta)+1+Q_{-i}\ge K.
\end{equation}
If the signer set is published, a signer receives private expressive value $e(x_i)\ge0$ from being named in the coalition. Unlike the article's constant $e$, expressive value here varies with the private signal: the function $e$ is continuous and strictly increasing, so a more favorable signal makes public endorsement more valuable. A nonsigner receives no payoff when someone else's coalition is published, a normalization that isolates the private return to joining. Giving nonsigners a public benefit would add a separate incentive to cast the pivotal endorsement. Signing costs $k>0$ whether the contract succeeds or fails, and under ideal failure privacy, failure creates neither expressive value nor exposure costs. A signer's payoff is therefore
\begin{equation}
u_i^{\mathrm{sign}}(Q_{-i},\theta;x_i,k)
=\ind\{b(\theta)+1+Q_{-i}\ge K\}e(x_i)-k,
\end{equation}
whereas not signing yields zero.

Under a symmetric cutoff $c$, every marginal agent signs exactly when her signal is at least $c$; by convention, an agent at the cutoff signs. Allowing $c\in\overline{\mathbb R}=\mathbb R\cup\{-\infty,+\infty\}$ includes ``always sign'' and ``never sign.'' If all other marginal agents use cutoff $c$, each of them signs in state $\theta$ with probability
\begin{equation}
p_c(\theta)
=1-F_\eta\!\left(\frac{c-\theta}{\sigma}\right).
\end{equation}
Use the conventions $p_{-\infty}=1$ and $p_{+\infty}=0$.
Conditional on $\theta$, the other $n-1$ agents make independent decisions with this signing probability. Their total is therefore $Q_{-i}(c)\sim\operatorname{Binomial}(n-1,p_c(\theta))$. After observing signal $x$, agent $i$ updates her belief about the state to the posterior density
\begin{equation}
g(\theta\mid x)
=\frac{g(\theta)f_\eta((x-\theta)/\sigma)}
{\int_{\underline\theta}^{\overline\theta}g(t)f_\eta((x-t)/\sigma)\,dt}.
\end{equation}
For a given state, let $n_{\mathrm{req}}(\theta)=\max\{0,K-b(\theta)-1\}$ be the smallest number of other marginal endorsements needed for publication after agent $i$ signs. The probability that the other agents supply at least this many endorsements is the upper tail of the binomial distribution:
\begin{equation}
\begin{aligned}
\omega_{\mathrm{pub}}(\theta;c)
&=\sum_{j=n_{\mathrm{req}}(\theta)}^{n-1}
\binom{n-1}{j}p_c(\theta)^j\{1-p_c(\theta)\}^{n-1-j},\\
n_{\mathrm{req}}(\theta)&=\max\{0,K-b(\theta)-1\},
\end{aligned}
\end{equation}
where the sum equals zero if $n_{\mathrm{req}}(\theta)>n-1$. Agent $i$ does not know $\theta$, so she averages this state-specific probability using her posterior belief. Her probability of success after observing $x$ is
\begin{equation}
p_i^{\mathrm{pub}}(x;c)
=\int_{\underline\theta}^{\overline\theta}
\omega_{\mathrm{pub}}(\theta;c)g(\theta\mid x)\,d\theta.
\end{equation}
The function $\omega_{\mathrm{pub}}(\theta;c)$ gives the chance of publication if the true state is $\theta$ and agent $i$ signs. The function $p_i^{\mathrm{pub}}(x;c)$ averages that chance over the states still possible after signal $x$; it is the finite model's counterpart of an attempt's success probability $p_i$ in the companion paper. Both vanguard participation $b(\theta)$ and each other marginal agent's signing probability $p_c(\theta)$ rise with the state, so $\omega_{\mathrm{pub}}(\theta;c)$ is weakly increasing in $\theta$. Monotone likelihood ratios make higher states more likely after a higher signal. Thus $p_i^{\mathrm{pub}}(x;c)$ is weakly increasing in $x$, and a signer's expected payoff is
\begin{equation}
U_i(x;c,k)=p_i^{\mathrm{pub}}(x;c)e(x)-k.
\end{equation}
Expected payoff equals the probability of publication times the value of publication, minus the signing cost.

The next assumption ensures that best responses take the form of cutoffs and that the symmetric cutoff is unique. It requires low signals at which not signing is always better, high signals at which signing is always better, and one point at which an agent is indifferent when everyone uses the same cutoff.

\begin{assumption}[Finite symmetric-cutoff regularity]\label{ass:finite_vanguard_selection}
There are nonempty intervals of low states for which $b(\theta)+n<K$ and high states for which $b(\theta)\ge K$. The first condition implies $n<K$. Thus, the marginal agents cannot publish without some vanguard participation, while the vanguard can publish by itself in favorable states. The success probability $p_i^{\mathrm{pub}}(x;c)$ is jointly continuous in $(x,c)$. For every cutoff $c\in\overline{\mathbb R}$, $U_i(\cdot;c,k)$ is continuous and strictly increasing on the relevant signal interval. Uniformly over $c\in\overline{\mathbb R}$ and $k$ in the compact interval $\mathcal K\subset(0,\infty)$, there are signals $x_L<x_H$ such that $U_i(x_L;c,k)<0<U_i(x_H;c,k)$. For every $k\in\mathcal K$, the diagonal payoff
\begin{equation}
\Phi_k(c)\equiv U_i(c;c,k)
\end{equation}
is continuous, has a unique root $c^*(k)$, and satisfies $\Phi_k(c)<0$ for $c<c^*(k)$ and $\Phi_k(c)>0$ for $c>c^*(k)$.
\end{assumption}

The prior and signal structure determine $p_i^{\mathrm{pub}}$ and ensure that a better signal cannot lower the chance of success. The rest of Assumption~\ref{ass:finite_vanguard_selection} restricts how an agent's preferred action changes as both her signal and her forecast of others' decisions change. The diagonal payoff $\Phi_k(c)=U_i(c;c,k)$ is the payoff of an agent with signal $c$ who expects every other marginal agent to use cutoff $c$. At a root of $\Phi_k$, the agent at the cutoff is indifferent, agents above it sign, and agents below it do not. The assumption requires exactly one such root and rules out repeated changes in the preferred action along the diagonal. The prior and vanguard do not guarantee this unique crossing; it is a maintained regularity condition, like the diagonal-crossing restrictions used in monotone global-games selection arguments \citep{frankelEquilibriumSelectionGlobal2003}.

\begin{proposition}[Finite cutoff characterization]\label{prop:finite_vanguard_selection}
Under Assumption~\ref{ass:finite_vanguard_selection}, the finite marginal-agent game has a unique equilibrium cutoff among symmetric monotone cutoff strategies. Under the stated boundary convention, marginal agent $i$ signs if and only if $x_i\ge c^*(k)$. Without that convention, cutoff strategies can differ only in the action assigned to the indifferent, zero-probability boundary signal. If $k_2>k_1$ are in $\mathcal K$, then $c^*(k_2)>c^*(k_1)$. Consequently, for every state $\theta$,
\begin{equation}
\begin{aligned}
&\Prob\!\left(b(\theta)+Q(c^*(k_2))\ge K\mid\theta\right)\\
&\qquad\le
\Prob\!\left(b(\theta)+Q(c^*(k_1))\ge K\mid\theta\right),
\end{aligned}
\end{equation}
where $Q(c)$ is the total number of marginal agents whose signals exceed $c$.
\end{proposition}

Within symmetric monotone strategies, the proposition identifies one equilibrium signal cutoff. A higher signing cost raises this cutoff: agents need stronger private information before signing, so fewer of them sign and publication cannot become more likely in any state. This finite result describes participation and publication within the contract. It neither compares the contract with simultaneous open expression nor models how coalition size changes a signer's exposure, so it does not extend the article's exposure-tail result to a finite population.

\begin{proof}
Fix a conjectured cutoff $c$. Continuity, strict monotonicity, and the uniform low- and high-signal inequalities imply that the best response is also a finite cutoff. The agent at that cutoff is indifferent. In a symmetric monotone cutoff equilibrium, the conjectured cutoff must equal the best-response cutoff, so the signal at the cutoff satisfies $U_i(c;c,k)=0$. Assumption~\ref{ass:finite_vanguard_selection} gives the unique solution $c^*(k)$, and the boundary convention tells the indifferent agent to sign. This establishes uniqueness only among symmetric monotone cutoff equilibria.

Set $c_1=c^*(k_1)$. Since $k_2>k_1$,
\begin{equation}
\Phi_{k_2}(c_1)=\Phi_{k_1}(c_1)-(k_2-k_1)<0.
\end{equation}
The assumed sign pattern for $\Phi_{k_2}$ then implies $c^*(k_2)>c_1$. Raising the cutoff weakly reduces the number of marginal signers for every realized set of private signals. Thus $b(\theta)+Q$ falls weakly \emph{pathwise}: it falls realization by realization, not just on average. Raising the cutoff cannot turn a failed realization into a success, which proves the publication-probability inequality in every state.
\end{proof}

\section{Appendix C. Failure Privacy, Attempt Privacy, and Exchangeability}\label{app:privacy_exchangeability}

Failure privacy makes one promise: if the publication condition is never met, no signer's identity is revealed. It does not by itself hide the failed count or the fact that a campaign was attempted. \emph{Attempt privacy} hides the campaign and its failure from the outside audience that can punish signers. A mechanism can therefore provide failure privacy while making the failed attempt public, but the ideal benchmark adds attempt privacy and withholds the failed count. If $V$ is the private signer set, the audience then observes
\begin{equation}
o=D_\tau(V)=
\begin{cases}
V, & m(V)\ge\tau,\\
\varnothing, & m(V)<\tau.
\end{cases}
\end{equation}
On success, the released signer set becomes the public coalition $C=D_\tau(V)$. On failure, the symbol $\varnothing$ is a public non-event: the audience sees no names, count, failure notice, or other evidence of an attempted campaign. Eligible agents may know that they received an invitation and may infer after the deadline that the contract failed. This private inference does not identify other signers or create a public signal for the audience that can punish them.

If the campaign is public, the absence of a published list reveals a failed attempt rather than the non-event $\varnothing$. Names remain private, but the audience can update its beliefs after learning that a campaign fell short. This is a different information regime. A breach or premature release is different again: it reveals information that the mechanism promised to hide.

The article assumes that coalitions of the same size impose the same exposure cost at a fixed audience-belief index. The following inference calculation adds \emph{informational exchangeability}: conditional on the state, coalitions of the same size are equally likely to form. Seeing Alice, Bob, and Carol instead of three other eligible people then conveys no information about the state beyond coalition mass. Equal exposure costs alone do not imply this additional restriction on what names reveal.

\begin{lemma}[Informational exchangeability]\label{lem:exchangeable_identities}
Suppose that, conditional on the state and the symmetric strategy, all coalitions of the same mass are equally likely. Suppose also that identities have no payoff-relevant public attributes beyond coalition membership. If publication reveals a named coalition $C$, then coalition mass contains all the information needed for public inference:
\begin{equation}
\Prob(\theta\mid C,\mathrm{publication})
=
\Prob(\theta\mid m(C),\mathrm{publication}).
\end{equation}
If $\mu(y)$ denotes the resulting audience belief at coalition size $y$, signer exposure costs can therefore be represented as
\begin{equation}
c_i(C)=\alpha_i\ell\bigl(m(C),\mu(m(C))\bigr).
\end{equation}
\end{lemma}

\begin{proof}
Under informational exchangeability, the likelihood of a particular published coalition $C$ depends on the state only through its mass $m(C)$. Bayes' rule therefore gives the same posterior belief about the state for every coalition of that mass. The cost representation follows because the audience responds only to coalition mass and the belief that mass induces.
\end{proof}

Publication also conveys information. Under a proper prior with finite mean, a nondegenerate upper-state publication event raises the expected state. When the state measures support, publication is good news about the statement. This publication-driven updating is distinct from direct safety in numbers at fixed $\mu$, where a larger coalition reduces exposure through $\ell(y,\mu)$ while audience beliefs are held fixed.

Lemma~\ref{lem:exchangeable_identities} identifies when coalition size contains all the information in the names on a published list. The publication-driven updating channel is outside the article's fixed-$\mu$ selection results, Lemma~\xref{lem:global_selection} and Theorem~\xref{thm:selection_gap}. If belief is instead $\mu(y)$, effective exposure is $\ell(y,\mu(y))$, and coalition size affects both protection and inference. Those results would then require the combined function to be bounded, continuous, and strictly decreasing, as well as the original dominance and single-crossing conditions. Exchangeability does not imply these properties. The appendix therefore makes no general selection claim when audience beliefs vary with coalition size.

Names still determine who becomes publicly associated with the statement and who bears exposure. Exchangeability removes only the effect of replacing one signer with another while holding coalition size fixed. If signer attributes change public beliefs, protection, or retaliation, coalition composition becomes another payoff-relevant state variable.

\section{Appendix D. Trust and Recantation}\label{app:trust_recantation}

The mechanism can fail in two different ways. First, the advertised contract may be false: a supposed administrator may collect names to leak them rather than follow the publication rule. That risk changes a prospective signer's expected payoff before she acts. Second, eligible people who sign in bad faith, or ``infiltrators'', may repudiate the statement as soon as the list becomes public, reducing the coalition that remains after release. The first problem calls for a trust threshold for signing; the second calls for a robust publication margin.

\subsection{A False Contract and Signer Trust}\label{app:false_contract}

Let $\delta_F\in[0,1)$ be the common probability that the apparent contract is false, and with probability $1-\delta_F$ the administrator is honest and follows the announced threshold rule. A false administrator gives each signer's identity to the opposition, never creates a protecting public coalition, and produces no expressive benefit. This is a conservative breach model: the opposition can target each signer, but the leaked names do not become a common public commitment. At state $s$, write
\begin{equation}
h_b(s)=\alpha\ell\bigl(b(s),\mu\bigr)
\end{equation}
for the resulting harm. Equation~\xeqref{eq:false_contract_payoff} gives the signer's expected payoff.

For any state $s$, integration over ranks gives
\begin{equation}
R_{SAC}^{F}(s)
=(1-\delta_F)R_{SAC}(s)
-\delta_F\bigl(h_b(s)+k\bigr).
\label{eq:false_contract_rank_relation}
\end{equation}
The baseline payoff $R_{SAC}$ already subtracts $k$. Multiplying it by $1-\delta_F$ would leave out $\delta_F k$ of the signing cost. The identity restores that cost because a signer pays it under either administrator type.

\begin{assumption}[False-contract selection domain]\label{ass:false_contract_selection}
Fix $\delta_F\in[0,1)$. Retain the exact-aggregation, prior, signal, $\tau$-protection, low-state exposure, genuine-failure-risk, vanguard-clearance, and strict-friction clauses of Assumption~\xref{ass:operative_domain}. Replace vanguard protectability with the risk-adjusted high-state condition
\begin{equation}
(1-\delta_F)e>h_b(\theta_H)+k.
\label{eq:false_contract_high_dominance}
\end{equation}
The payoff in equation~\xeqref{eq:false_contract_payoff} is then bounded and nondecreasing in the state, and $R_{SAC}^{F}$ is continuous and strictly single crossing on $(\theta_L,\theta_H)$, with a unique zero $\theta_F^*$.
\end{assumption}

To interpret the high-state condition, suppose the vanguard clears the threshold and no marginal agent joins. An honest administrator yields $e-h_b(\theta_H)$, a false administrator yields $-h_b(\theta_H)$, and signing always costs $k$. Expected payoff is therefore $(1-\delta_F)e-h_b(\theta_H)-k$. At a low state where even full participation cannot reach the threshold, expected payoff is $-\delta_F h_b(\theta_L)-k<0$. Under the article's signal-concentration requirements, these strict state-level inequalities give agents sufficiently favorable signals at which signing dominates and sufficiently unfavorable signals at which not signing dominates.

The baseline's unique crossing does not guarantee the crossing clause here: equation~\eqref{eq:false_contract_rank_relation} changes the state-dependent rank payoff. A sufficient condition is that the baseline $R_{SAC}$ be strictly increasing wherever it is nonnegative on the compact selection interval $\mathcal{J}$. This allows the guaranteed-failure region, where $R_{SAC}=-k$. The baseline payoff is weakly increasing throughout $\mathcal{J}$, and $h_b(s)$ is weakly decreasing because the vanguard grows and exposure falls with coalition size. Wherever $R_{SAC}<0$, the false-contract payoff is also negative. On the remaining interval, equation~\eqref{eq:false_contract_rank_relation} is strictly increasing because $1-\delta_F>0$. Continuity and the low- and high-state inequalities therefore give it exactly one zero.

Under Assumption~\ref{ass:false_contract_selection}, the proof of Lemma~\xref{lem:global_selection} applies without another symmetry or cutoff restriction. Below the publication threshold, payoff is constant in marginal participation. At the threshold it jumps by $(1-\delta_F)(e-\alpha\ell(\tau,\mu))\ge0$, and above the threshold it rises as exposure falls. The false-contract arm therefore retains strategic complementarity.

\begin{corollary}[False-contract noise-limit selection]\label{cor:false_contract_selection}
Under Assumption~\ref{ass:false_contract_selection}, the lower and upper signal boundaries of the strategies surviving iterated deletion for the false-contract arm both converge to $\theta_F^*$ as $\sigma\downarrow0$. For every $\varepsilon>0$, sufficiently small noise forces participation at every signal above $\theta_F^*+\varepsilon$ and nonparticipation at every signal below $\theta_F^*-\varepsilon$. At the fixed interior evaluation state, the corresponding conditional expressive-benefit requirement across one-type economies is $e_{SAC}^{sel,F}(\theta)$ in Proposition~\xref{prop:trust_recantation}.
\end{corollary}

\begin{proof}
The proof of Lemma~\xref{lem:global_selection} applies to any arm that satisfies its conditions. Assumption~\ref{ass:false_contract_selection} supplies the required prior, signal, aggregation, state-level dominance, and unique-crossing conditions for $u_{SAC}^{F}$, and signal dominance follows as in Appendix~A because the payoff is bounded, nondecreasing in the state, and strictly signed at and beyond the dominance states. The preceding paragraph verifies strategic complementarity. Substituting $u_{SAC}^{F}$ and $R_{SAC}^{F}$ into that proof gives convergence to the unique zero $\theta_F^*$. At the evaluation state, the required expressive benefit is the larger of the rank-indifference value and the value that protects signers at the publication threshold. The result is asymptotic: at positive noise the boundaries of the surviving strategies need not coincide, and exact rank indifference does not prescribe a unique action.
\end{proof}

\begin{proof}[Proof of Proposition~\getrefnumber{prop:trust_recantation}, part 1]
Fix the article's interior evaluation state, abbreviate $q_\tau=1-w_\tau$, and let
\begin{equation}
A_\tau
=\alpha\int_{w_\tau}^{1}\ell\bigl(y(w;\theta),\mu\bigr)\,dw.
\end{equation}
The false-contract rank payoff is
\begin{equation}
(1-\delta_F)(q_\tau e-A_\tau)-\delta_F h_b(\theta)-k.
\end{equation}
Setting it equal to zero gives
\begin{equation}
\bar e_{SAC}^{F}(\theta)
=\frac{A_\tau}{q_\tau}
+\frac{\delta_F h_b(\theta)+k}{(1-\delta_F)q_\tau}.
\end{equation}
Because $\bar e_{SAC}=A_\tau/q_\tau+k/q_\tau$, subtraction yields equation~\xeqref{eq:false_contract_rank_value}. Hence $\bar e_{SAC}^{F}=\bar e_{SAC}$ at $\delta_F=0$, and
\begin{equation}
\frac{\partial\bar e_{SAC}^{F}}{\partial\delta_F}
=\frac{h_b(\theta)+k}{q_\tau(1-\delta_F)^2}>0.
\end{equation}
The protection condition is unchanged under an honest administrator: publication at the threshold must satisfy $e\ge\alpha\ell(\tau,\mu)$. Rank indifference and protection therefore give the maximum in the proposition. Corollary~\ref{cor:false_contract_selection} translates this rank requirement into conditional noise-limit selection across the one-type economies.

Now fix a type that satisfies $\tau$-protection with $e>\bar e_{SAC}$ and define $d=q_\tau(e-\bar e_{SAC})>0$. The condition $e\ge\bar e_{SAC}^{F}$ is equivalent to
\begin{equation}
(1-\delta_F)d\ge\delta_F\bigl(h_b(\theta)+k\bigr).
\end{equation}
Rearranging gives
\begin{equation}
1-\delta_F\ge
\frac{h_b(\theta)+k}{h_b(\theta)+k+d}
=t_{sign}^{*}(\theta;e).
\end{equation}
The inequality is strict for strictly positive selected rank payoff and becomes equality at indifference.

For the institutional comparison, assume $G^{acct}>0$ and $\bar e_{open}>\alpha\ell(\tau,\mu)$, so protection does not itself eliminate a positive selected gap. The gap is positive precisely when $\bar e_{SAC}^{F}<\bar e_{open}$, or
\begin{equation}
G^{acct}>
\frac{\delta_F\bigl(h_b(\theta)+k\bigr)}{(1-\delta_F)q_\tau}.
\end{equation}
Rearrangement gives equation~\xeqref{eq:advantage_trust_threshold}. Write either threshold as $(h_b+k)/(h_b+k+X)$, where $X=d$ for $t_{sign}^{*}$ and $X=q_\tau G^{acct}$ for $t_{adv}^{*}$; in both cases $X>0$. Holding $X$ and $k$ fixed, the partial derivative with respect to $h_b$ is $X/(h_b+k+X)^2>0$. Greater leak harm therefore raises minimum trust in this comparison. The two thresholds generally differ because $t_{sign}^{*}$ asks whether one fixed type participates, while $t_{adv}^{*}$ asks whether the institution retains an advantage over open expression.
\end{proof}

If a leaked signer list itself creates some protection, replace $h_b(\theta)$ in the rank calculation by the exposure harm under a false administrator, averaged over ranks. That average is no greater than $h_b(\theta)$ when the full harvested list is visible and safety in numbers holds. The conservative convention of protection from the vanguard alone therefore requires weakly more trust.

\subsection{Publication Against Bounded Recantation}\label{app:bounded_recantation}

Authentication proves that an endorsement came from one eligible person and applies to the announced statement, but it does not prove continuing support. This extension therefore relaxes the baseline restriction $V\subseteq S$ and permits an authenticated signer set $V\subseteq N$ to include infiltrators. The calculation is nonatomic: authenticated endorsements can attain any mass up to the maximum capacity $\widehat Y$. The uncertainty set permits every recanting subset up to the stated absolute or proportional bound, including a subset that attains the bound whenever enough mass is available.

The administrator first releases the historical signer list $V$, and any recantations then become public before the opposition responds. If $Z\subseteq V$ is the recanting set, the coalition that continues to stand behind the statement is $C^{stay}=V\setminus Z$. A continuing signer's exposure is $\alpha\ell(m(C^{stay}),\mu)$. Thus a recanter remains visible on the historical list but does not supply protection to those who continue to endorse it. For $0<\rho\le1$, a publication rule is robust to the stated bound if every allowed $Z$ leaves $m(C^{stay})\ge\rho$.

\begin{proof}[Proof of Proposition~\getrefnumber{prop:trust_recantation}, part 2]
First suppose the mass that can recant is bounded by $\bar z$. If a threshold rule publishes only when $m(V)\ge\rho+\bar z$, then
\begin{equation}
m(C^{stay})
=m(V)-m(Z)
\ge m(V)-\bar z
\ge\rho.
\end{equation}
Thus $\tau_{\bar z}=\rho+\bar z$ is robust. It is sharp among size-only threshold rules under the stated uncertainty set. At any attainable release mass $y<\rho+\bar z$, choose an allowed recanting mass $\min\{\bar z,y\}$. The continuing mass is $\max\{0,y-\bar z\}<\rho$. No lower attainable threshold guarantees safety. If $\rho+\bar z>\widehat Y$, no nontrivial size-only rule can both release and provide the guarantee; always withholding remains safe.

For the proportional bound, let $V_0\subseteq S$ be a known vanguard whose endorsements remain committed, with fixed mass $m(V_0)=b_0<\rho$. Suppose recantation is restricted to $V\setminus V_0$ and satisfies
\begin{equation}
m(Z)\le\phi\bigl(m(V)-b_0\bigr),
\qquad \phi\in[0,1).
\end{equation}
If the released mass is $y$, the continuing coalition has mass at least
\begin{equation}
b_0+(1-\phi)(y-b_0).
\end{equation}
Setting this expression weakly above $\rho$ and solving for $y$ gives
\begin{equation}
y\ge b_0+\frac{\rho-b_0}{1-\phi}=\tau_\phi.
\end{equation}
Subtracting $\rho$ yields the displayed margin in equation~\xeqref{eq:proportional_recantation_margin}. Publishing at any smaller attainable mass permits the maximal allowed fraction to recant and leaves less than $\rho$, so the bound is sharp. If $\tau_\phi>\widehat Y$, only withholding provides the robust guarantee.
\end{proof}

The proportional formula is a fixed precommitted threshold because $b_0$ is a known nonrecanting mass, not the state-dependent vanguard $b(\theta)$. If the architect knows only that $b(s)\ge\underline b$ throughout the operative state interval, the same guarantee follows by setting $b_0=\underline b$. The formula differs from $\rho/(1-\phi)$ because only endorsements beyond the genuine vanguard are discounted. If the same fraction could recant from the entire released coalition, including the vanguard, then $b_0=0$ and $\rho/(1-\phi)$ would be the correct threshold.

This technology is intentionally separate from the article's baseline, in which exposure depends on the mass of the original named coalition $C=D(V)$. If a repudiated name continues to consume the opposition's retaliation capacity, recantation does not reduce protective mass and no extra margin is needed. If current commitment determines protection, a participation model must replace successful-rank exposure by $\ell(y-\bar z,\mu)$ in the absolute worst case or by $\ell(b_0+(1-\phi)(y-b_0),\mu)$ under the proportional bound. That change affects the full rank integral, so Theorem~\xref{thm:selection_gap} cannot be imported by changing $\tau$ alone.

The recantation result is a worst-case design guarantee, not an equilibrium model of infiltration. It does not determine how many infiltrators enter, what they pay, or how the architect learns $\bar z$ or $\phi$. Given a credible bound and protection supplied only by continuing commitment, it identifies the necessary and sufficient publication condition. Screening, penalties, and governance can lower the bound; the margin keeps the continuing coalition safe when those controls are imperfect.

\section{Appendix E. Implementation with Existing Tools}\label{app:implementation}

A fixed-population Social Assurance Contract must perform four functions. It must verify that each eligible person submits at most one endorsement, keep identities and endorsements in private custody during collection, determine whether the publication condition has been met, and jointly publish the complete authenticated signer set after success. Existing systems already perform close versions of all four functions. A proof-of-concept threshold secret petition accepts encrypted signatures from validated users, prevents duplicate signatures, divides decryption control among trustees, and reveals the signatures after the threshold is reached \citep{breuerSeCritMassThresholdSecret2024}. Secure allegation escrows divide confidential allegations and identities among independent parties, then reveal them after specified matching thresholds are reached; these systems have formal security analyses and working prototypes \citep{arunFindingSafetyNumbers2020}. \citet{mangipudiCollusionDeterrentThresholdInformation2023} develop a related threshold information escrow that deters trustee collusion. The necessary components therefore have close precedents, but building a production system requires combining them, testing the result for security and usability, and defining clear operating rules.

\subsection{A Feasible Architecture}\label{app:implementation_architecture}

The central design choice is whether to trust one administrator or divide administration among several actors. A single administrator is easier to operate; a distributed architecture limits what any one actor can learn or do.

One independent custodian, such as a law firm, regulator, union, or ombudsperson, can perform all four functions. The custodian checks participants against a roster, protects the list with access controls, and follows audited rules for release and destruction. Professional duties, contractual penalties, independent audits, and reputation can make this arrangement credible, but participants must trust one organization that can inspect, alter, suppress, or leak the list.

A distributed design separates five roles. The \emph{architect} fixes the contract but never receives submissions. A \emph{registrar} verifies eligibility, and a \emph{collection service} checks and counts encrypted submissions. An independent \emph{auditor} resolves disputes over whether a valid submission was included, and $m$ independent \emph{trustees} control decryption. The registrar, collection service, auditor, and trustees together perform the administrator's functions. The trustees jointly create an encryption key for that campaign, but each trustee keeps only one share; no central operator holds the complete key. At least $t$ trustees must cooperate to decrypt the records, where $1<t\le m$. Different existing institutions can fill these roles, and dividing them prevents one actor from possessing both the identities and unilateral power to release them.

Each cryptographic tool solves a different problem. A digital signature, the electronic equivalent of signing a particular document, proves that the holder of a key endorsed the exact statement and that the statement has not changed. Encryption turns the signed record into unreadable ciphertext. Threshold decryption divides the power to reverse that encryption among several trustees, so no trustee can reveal a record alone. A zero-knowledge proof lets the collector verify a claim about the encrypted record---for example, that it contains an eligible identity and a valid signature---without seeing the identity or signature. Together, these tools provide authenticity, secrecy, shared control over release, and verification without disclosure.

\begin{enumerate}
\item Before recruitment begins, the architect fixes the exact statement, the eligible population $N$, the publication threshold $K$, the closing time, and what the system will disclose after success or failure. The architect also fixes any cancellation and emergency rules. A digital signature records assent to these terms, and encryption hides the signed record during collection. Both are necessary: one proves agreement, while the other provides secrecy.

\item The registrar verifies each member of $N$ once and gives her a cryptographically verifiable eligibility credential. The credential certifies either her identity or a public signing key that only she controls and that the registrar has linked to her identity. People can register before any campaign exists, so registration does not reveal an intention to sign a particular statement. An anonymous credential lets a participant prove eligibility to the collector without revealing her identity during collection. Established systems also support one-use credentials that reveal identity later under specified conditions \citep{camenischEfficientSystemNontransferable2001}.

\item A participant uses her identity-bound key to sign the exact statement and a unique identifier for the contract. She submits an encrypted record containing the signature, her certified identity, and the identifier. She also submits a campaign-specific serial number and a zero-knowledge proof. The proof shows that the encrypted identity or key was certified by the registrar and that the record contains a valid signature on this contract, without revealing the identity or signature. Anonymous-credential and private-contract systems provide the proof techniques needed for these checks \citep{camenischEfficientSystemNontransferable2001,kosbaHawkBlockchainModel2016}. The same credential always produces the same serial number within one campaign but an unrelated number in another campaign. The collector can therefore reject a duplicate without learning the participant's identity or linking her activity across campaigns. Coconut implements this eligibility and duplicate-prevention pattern for electronic petitions \citep{sonninoCoconutThresholdIssuance2019}.

\item The collection service verifies the proof, rejects repeated serial numbers, and records each accepted encrypted submission in a private authenticated log. Authentication makes later deletion or alteration detectable, and the service gives the participant a signed inclusion receipt. A participant who receives no receipt, or whose receipt is missing from the count, can challenge the omission privately before an independent auditor. With a receipt, she can contest suppression without publicly identifying herself in a failed campaign.

\item Vanguard members use the same authenticated process, and their accepted endorsements count toward $K$. They participate independently of the marginal participants' choices. As in the finite benchmark, the service reveals neither the interim vanguard count nor the total count to marginal participants. At the deadline, the collector and trustees audit the accepted log. If the count itself must remain secret, secure multiparty computation can test whether the private inputs reach $K$ without giving every input to one operator \citep{evansPragmaticIntroductionSecure2018}.

\item At the announced deadline, each trustee verifies the integrity and count of the same private log. The auditor resolves any timely challenge supported by an inclusion receipt. If the count reaches $K$, at least $t$ trustees sign a success certificate that links the contract identifier, deadline, final count, and a cryptographic fingerprint of the accepted ciphertexts. The fingerprint is a short value used to detect changes to those encrypted records, tying the certificate to the set the trustees approved. The same trustee quorum then combines its decryption-key shares and jointly publishes the complete authenticated signer set as coalition $C=D_K(V)$. Any reader can verify both the registrar's link between identity and key and the participant's signature on the contract.

The publication threshold $K$ and trustee quorum $t$ serve different purposes. The first authorizes publication; the second specifies how many trustees must cooperate to make the records readable.

If the count falls short, the administrator publishes no identities, the collector deletes all identity-bearing records, and the trustees erase the campaign-specific key shares. Failure privacy requires that fewer than $t$ trustees collude before erasure and that no group of $t$ trustees retain enough shares to reconstruct the key later.
\end{enumerate}

This workflow creates no public transaction, log entry, or counter for an individual submission, and it issues no public failure notice. Only success produces a public certificate and named coalition. For a failed campaign to remain the article's ideal public non-event $\varnothing$, recruitment and network traffic must also provide attempt privacy: they must hide the campaign from the audience that can punish signers.

Several established technologies can perform parts of this workflow. Threshold cryptography divides decryption control: no trustee can open a record alone, but a specified quorum can combine its shares. A smart contract is program code that a blockchain executes and records under rules fixed in advance. Privacy-preserving voting systems show that smart contracts can enforce eligibility, verify proofs, manage stages, and tally votes for a fixed group \citep{mccorrySmartContractBoardroom2017}. More general systems use encryption and zero-knowledge proofs to enforce private contract rules without placing the hidden information in plain text on a public ledger \citep{kosbaHawkBlockchainModel2016}.

\paragraph{Public-counter variant}
A public ledger, a record copied across many computers and designed to make past entries hard to alter, can store encrypted records or cryptographic fingerprints of them. Encryption hides a record's contents; a fingerprint makes changes detectable. A fingerprint does not contain the record itself, so copies must also be stored somewhere to provide \emph{availability}: the ability to retrieve the record when release is authorized. Recording each submission publicly reveals that a campaign exists and may reveal its running count. This protects failed names but does not create the article's stronger public non-event unless fixed-size batches, decoy traffic, or another privacy layer also hides the traffic and count. An immutable ledger that stores identity-bearing ciphertexts also preserves material that a future key compromise could expose. The benchmark-compatible design therefore keeps those ciphertexts off the public ledger. The ledger can instead record the terms fixed in advance and, after success, the certificate and released list. A public-ledger design also depends on the security and availability of its consensus system, the rules through which participating computers agree on the valid state of the ledger.

In this design, the registrar knows who is eligible but not who submitted, and the collector verifies that submissions are unique without reading them. No single trustee can open a failed list, and the architect cannot change the statement or threshold after seeing participation. Submission receipts let a private auditor detect omitted entries or different versions of the log shown to different participants. The design still depends on reliable operation. Credentials must be issued correctly or audited; participant devices and software must work as intended; and all trustees must use the same log. At least $t$ trustees must remain available after success, fewer than $t$ may collude before success, and key shares must be erased after failure. Each deployment must enforce and test these conditions through engineering controls, audits, and governance rules.

\subsection{Threats and Mitigations}\label{app:implementation_threats}

Existing components show that the core functions are feasible. A working design must also account for false identities and malicious participants, custodians who collude or become unavailable, compromised software and governance, and observable metadata.

\paragraph{Eligibility and malicious participation}
A Sybil attack occurs when one actor uses several identities and is counted several times \citep{douceurSybilAttack2002}. A fixed eligible population makes this threat easier to control than it is on the open internet, where the system must determine whether every account belongs to a different person. The architect fixes the eligibility rule, the registrar gives one credential to each verified person, and campaign-specific serial numbers let the collector reject repeated use. Auditable credential issuance, strong authentication, revocation of stolen credentials, and separation of registration from collection address the main risks. These controls prevent impersonation and duplicate counting, but authentication proves identity rather than motive. An eligible infiltrator can still sign to push the list over the threshold and then recant. Proposition~\xref{prop:trust_recantation} shows how a sharp margin above the minimum safe size protects the remaining coalition when recanting mass is bounded; it does not solve the infiltrator's entry decision.

\paragraph{Custody, collusion, and availability}
A centralized custodian can leak names, invent or discard entries, suppress a successful list, or submit to coercion. Distributed trustees, signed submission receipts, an append-only log, and a quorum release rule divide these powers. But enough trustees may cooperate to open records early, while too few available trustees may block release after success. Choosing trustees from institutions with different interests, setting a quorum that tolerates some unavailable members, protecting key shares, and fixing legal duties and penalties in advance reduce both risks. Collusion-deterrent escrow protocols can also make premature cooperation less attractive to trustees \citep{mangipudiCollusionDeterrentThresholdInformation2023}. The false-contract calculation in Appendix~\ref{app:false_contract} gives the minimum trust a signer needs to participate when she still doubts the advertised custody arrangement.

\paragraph{Software, keys, and governance}
Sound cryptography does not protect against compromised devices, poor key storage, software bugs, or an administrator who can change the program. A participant's device can reveal her action or trick her into signing a different statement, and a bug can reject valid submissions or apply the wrong threshold. Conservative software design, independent security audits, open-source review, hardware-backed or multiparty key custody, and clear recovery procedures reduce these risks. The contract identifier should cryptographically bind the statement, roster rule, threshold, deadline, and release policy so that one administrator cannot change them after collection begins. Related infrastructure already operates at large scale: decentralized-finance smart contracts held about \$150 billion in total value as of April 2022 \citep{wernerSoKDecentralizedFinance2022}. Attacks on these systems have also revealed recurring failures in design, governance, dependencies, and implementation \citep{rezaeiSoKRootCause2025}. Those failures supply a concrete checklist for design and audit.

\paragraph{Metadata and public observability}
Encryption hides a submission's contents, but it does not automatically hide its \emph{metadata}: where it came from, when it was sent, or even that it exists. If every encrypted submission appears as a public transaction, observers may infer the interim count or link a transaction to a participant without reading its contents. The system can reduce these disclosures by grouping submissions into batches, routing them through intermediaries, delaying their recording, or collecting them privately instead of on a public chain. Failure privacy can coexist with a public encrypted count, but the article's stronger observation $\varnothing$ requires attempt privacy: the audience that can impose exposure costs sees no count, failure notice, or evidence of an attempted campaign.

Legal and institutional controls complement cryptography. Contracts can penalize misconduct, audits can expose violations of the release rule, and transparent open-source operation can build a strong reputation. Instead of trusting one custodian with the entire system, participants in the distributed design rely on narrower and auditable claims about credential issuance, software, logs, and a trustee quorum.

\subsection{Networks and Deployment}\label{app:implementation_networks}

Cryptographic security does not by itself create a viable campaign. The architect and administrator must reach potential signers, earn their trust, protect them from the relevant audience, and control who learns that the campaign exists. Four different networks govern these tasks; they need not contain the same people or connect them in the same way. The \emph{recruitment network} determines which eligible people hear about the contract, and the \emph{trust network} determines whether they believe the invitation and the custody arrangement. The \emph{exposure network} identifies who can impose material costs after names become public, and the \emph{observation network} determines whether outsiders learn that a campaign was attempted or failed.

Dense, trusted groups can supply the vanguard. An architect can first secure commitments from protected or already-public supporters through professional, organizational, or social ties, then approach more distant eligible participants. Connections across departments or social circles determine how far recruitment can spread. In the model, these features determine vanguard mass $b(\theta)$, participation capacity $Y(\theta)$, and whether the threshold is reachable. A fragmented recruitment network reduces capacity even when many people privately support the statement.

The exposure network identifies who can punish a signer. Inside an organization, the relevant audience may be supervisors or close colleagues, while in a public campaign it may include employers, professional associations, or a wider political audience. The article's homogeneous exposure benchmark fits settings in which protection depends mainly on final coalition size. When network position, seniority, or subgroup membership changes protection, the heterogeneous technology in the companion paper is the relevant model.

Private recruitment and public observability are separate design choices. Trusted invitations can spread through the recruitment network while the outside audience sees no campaign. Public advertising or a visible progress counter may reach more people, but it also makes failure public and may reveal how close the campaign came to success. The benchmark therefore describes private recruitment around a credible vanguard, enough potential participation to reach the threshold in the state being studied, a vanguard that can clear the threshold in favorable states, and an information policy that hides failed lists and attempts from the exposure audience.

Implementation burden and distrust raise the signing friction $k$ by making participation harder or less attractive. Premature or partial disclosure is more damaging: it exposes signers before the coalition is large enough to protect them and breaks the mechanism studied in the paper. A working system must provide eligibility checks, secure custody, verifiable counting, and conditional release against the attacks it faces. Existing institutional and cryptographic tools provide close versions of these functions, but deployment security depends on how well they are combined and operated.

\ifSubfilesClassLoaded{%
  \bibliographystyle{aea}
  \bibliography{conditional_disclosure}
}{ }

\end{document}

\fi

%%%%%%%%%%%%%%%%%
\end{document}
%%%%%%%%%%%%%%%%%